\documentclass[acmsmall,nonacm]{acmart}     

\setcopyright{rightsretained}
\copyrightyear{2019}
\acmYear{2019}
\acmMonth{09}
\acmDOI{}
\acmISBN{}

\acmJournal{FACMP}
\acmVolume{1}
\acmNumber{1}
\acmArticle{1}

\usepackage[T1]{fontenc}
\usepackage{microtype}
\usepackage{setspace}
\usepackage{pdfpages}

\usepackage[framemethod=tikz]{mdframed}

\usepackage{xspace}
\usepackage{relsize} 

\usepackage{stmaryrd} 
\usepackage{mathtools}
\usepackage{stackrel}

\usepackage[group-digits=integer, group-minimum-digits=4, detect-weight=true,
                list-final-separator={, and }, add-integer-zero=false,
                free-standing-units, unit-optional-argument, 
                binary-units,
                ]{siunitx}

\usepackage{balance}
\usepackage{paralist}
\usepackage{url}
\usepackage{enumitem}

\renewcommand{\subsubsection}[1]{\smallskip\noindent{\em\bfseries #1.}}

\usepackage{longtable}
\usepackage{booktabs}
\usepackage{multirow}
\usepackage{rotating} 

\usepackage{tabulary}

\usepackage{dcolumn}
\newcolumntype{d}{D{.}{.}{3.0}}
\newcolumntype{j}{D{.}{.}{2.0}}
\newcolumntype{k}{D{.}{.}{1.0}}
\newcolumntype{e}{D{.}{.}{4.0}}
\newcolumntype{f}{D{.}{.}{5.0}}
\newcolumntype{g}{D{.}{.}{4.1}}

\newcommand{\PreserveBackslash}[1]{\let\temp=\\#1\let\\=\temp}
\newcolumntype{C}[1]{>{\PreserveBackslash\centering}p{#1}}
\newcolumntype{R}[1]{>{\PreserveBackslash\raggedleft}p{#1}}
\newcolumntype{L}[1]{>{\PreserveBackslash\raggedright}p{#1}}

\usepackage{subcaption} 

\usepackage{tikz}
\usepackage{pgflibraryshapes}
\usetikzlibrary{arrows}
\usetikzlibrary{chains}
\usetikzlibrary{fadings}
\usetikzlibrary{matrix}
\usetikzlibrary{fit}
\usetikzlibrary{positioning}
\usetikzlibrary{shapes}
\usetikzlibrary{automata}

\usepackage{pgfplots}
\pgfplotsset{compat=1.8}
\usepgfplotslibrary{statistics}

\usepackage[noend]{algorithmic}
\usepackage{algorithm}
\usepackage{listings}
\lstdefinestyle{C}{
    language=C,
    basicstyle=\sffamily\small,
    numberblanklines=true,
    columns=fixed,
    aboveskip=2pt,
    belowskip=1pt,
    lineskip=0pt,
    numbers=left,
    numberstyle=\tiny,
    stepnumber=1,
    numberfirstline=true,
    firstnumber=1,
    xleftmargin=15pt,
    keywordstyle=\color{black},
    commentstyle=\color{gray},
    morekeywords={assert},
}

\usepackage[rounded]{syntax}
\let\oldsdlengths\sdlengths
\renewcommand\sdlengths{\oldsdlengths\setlength{\sdfinalskip}{0pt}}

\usepackage{diagbox}

\usepackage{array}
\newcolumntype{C}[1]{>{\centering}m{#1}}

\usepackage{datenumber}

\definecolor{quotemark}{gray}{0.7}
\makeatletter
\def\fquote{%
    \@ifnextchar[{\fquote@i}{\fquote@i[]}
           }%
\def\fquote@i[#1]{%
    \def\tempa{#1}%
    \@ifnextchar[{\fquote@ii}{\fquote@ii[]}
                 }%
\def\fquote@ii[#1]{%
    \def\tempb{#1}%
    \@ifnextchar[{\fquote@iii}{\fquote@iii[]}
                      }%
\def\fquote@iii[#1]{%
    \def\tempc{#1}%
    \vspace{1em}%
    \noindent%
    \begin{list}{}{%
         \setlength{\leftmargin}{0.1\textwidth}%
         \setlength{\rightmargin}{0.1\textwidth}%
                  }%
         \item[]%
         \begin{picture}(0,0)%
         \put(-15,-5){\makebox(0,0){\scalebox{3}{\textcolor{quotemark}{``}}}}%
         \end{picture}%
         \begingroup\itshape}%
 \def\endfquote{%
 \endgroup\par%
 \makebox[0pt][l]{%
 \hspace{0.8\textwidth}%
 \begin{picture}(0,0)(0,0)%
 \put(15,15){\makebox(0,0){%
 \scalebox{3}{\color{quotemark}''}}}%
 \end{picture}}%
 \ifx\tempa\empty%
 \else%
    \ifx\tempc\empty%
       \hfill\rule{100pt}{0.5pt}\\\mbox{}\hfill\tempa%
   \else%
       \hfill\rule{100pt}{0.5pt}\\\mbox{}\hfill\tempa,\ \emph{\tempb},\ \tempc%
   \fi\fi\par%
   \vspace{0.5em}%
 \end{list}%
 }%
 \makeatother

\usepackage{todonotes}
\newcommand{\astodo}[1]{}
\newcommand{\asnext}[1]{}

\newcommand{\recap}[1]{}

\newcommand{\important}{}

\usepackage[many]{tcolorbox}
\newtcolorbox{myopbox}[1][]{
    breakable,
    title=#1,
    colback=white,
    colbacktitle=white,
    coltitle=black,
    fonttitle=\bfseries,
    bottomrule=0pt,
    toprule=0pt,
    leftrule=1pt,
    rightrule=1pt,
    titlerule=0pt,
    arc=0pt,
    outer arc=0pt,
    colframe=gray,
}
\newtcolorbox{recapbox}[1][]{
    breakable,
    colback=gray,
    colbacktitle=white,
    coltitle=black,
    fonttitle=\bfseries,
    bottomrule=0pt,
    toprule=0pt,
    leftrule=0pt,
    rightrule=0pt,
    titlerule=0pt,
    arc=0pt,
    outer arc=0pt,
    colframe=gray,
}
\newtcolorbox{keywordsbox}[1][]{
    breakable,
    title=#1,
    colback=white,
    colbacktitle=white,
    coltitle=black,
    fonttitle=\bfseries,
    bottomrule=1pt,
    toprule=1pt,
    leftrule=0pt,
    rightrule=0pt,
    titlerule=0pt,
    arc=0pt,
    outer arc=0pt,
    colframe=black,
}
\newtcolorbox{storrybox}[1][]{
    breakable,
    title=#1,
    colback=white,
    colbacktitle=white,
    coltitle=black,
    fonttitle=\bfseries,
    bottomrule=0pt,
    toprule=0pt,
    leftrule=0pt,
    rightrule=0pt,
    titlerule=0pt,
    arc=0pt,
    outer arc=0pt,
    colframe=black,
}

\newtcolorbox{defbox}{
  skin=enhanced,
  breakable,
  colback=white,
  colframe=gray!30,
  coltitle=black,
  fonttitle=\bfseries,
  leftrule=0.5cm,
  arc=0pt,
  outer arc=0pt,
  title={Definition:}
}

\newtcolorbox{parabox}[1][]{
    breakable,
    title=#1,
    colback=gray!30,
    colbacktitle=white,
    coltitle=black,
    fonttitle=\bfseries,
    bottomrule=0pt,
    toprule=0pt,
    leftrule=0pt,
    rightrule=0pt,
    titlerule=0pt,
    arc=0pt,
    outer arc=0pt,
    colframe=black,
}

\tikzset{initial text={}}

\newcommand{\cA}{1}
\newcommand{\cB}{0}

\newcommand{\perfshift}[2]{%
\renewcommand{\cA}{#1}%
\renewcommand{\cB}{#2}%
\begin{tikzpicture}
    \pgfmathparse{\cB*-1} \pgfmathresult \let\paramMin\pgfmathresult
    \pgfmathparse{\cB} \pgfmathresult \let\paramMax\pgfmathresult
    \pgfmathparse{\cA} \pgfmathresult \let\paramValue\pgfmathresult
   
    \begin{axis}[hide axis, scale only axis, height=2ex, width=10ex,
    restrict y to domain=\paramMin:\paramMax]    
    \addplot +[mark size=1pt, fill=black, draw=black, color=black, black] coordinates {(\paramValue, 0)};
    \addplot [black] coordinates { (\paramMin,0) (\paramMax,0) };
    \addplot [black] coordinates { (0,1) (0,-1) };

    \end{axis}
\end{tikzpicture}
}

\newcommand{\perfshiftlog}[2]{%
\renewcommand{\cA}{#1}%
\renewcommand{\cB}{#2}%
\begin{tikzpicture}
    \pgfmathparse{\cB*-1} \pgfmathresult \let\paramMin\pgfmathresult
    \pgfmathparse{\cB} \pgfmathresult \let\paramMax\pgfmathresult
    \pgfmathparse{\cA} \pgfmathresult \let\paramValue\pgfmathresult
   
    \begin{axis}[hide axis, scale only axis, height=2ex, width=10ex,
    restrict y to domain=\paramMin:\paramMax,xmode=log]    
    \addplot +[mark size=1pt, fill=black, draw=black, color=black, black] coordinates {(\paramValue, 0)};
    \addplot [black] coordinates { (\paramMin,0) (\paramMax,0) };
    \addplot [black] coordinates { (0,1) (0,-1) };

    \end{axis}
\end{tikzpicture}
}

\newcommand\definetool[2]{\newcommand{#1}{{\textsf{#2}}\xspace}}
\definetool{\framac}     {Frama-C}
\definetool{\blast}     {Blast}
\definetool{\varvel}     {Varvel}
\definetool{\orion}     {Orion}
\definetool{\cpachecker}{CPAchecker}
\definetool{\smtinterpol}{SMTInterpol}
\definetool{\mathsat}{MathSAT5}
\definetool{\cbmc}      {Cbmc}
\definetool{\esbmc}     {ESBMC}
\definetool{\cil}       {Cil}
\definetool{\llvm}      {LLVM}
\definetool{\tvla}      {Tvla}
\definetool{\ocaml}     {OCaml}
\definetool{\tvp}       {Tvp}
\definetool{\camplp}    {CamlP4}
\definetool{\foci}      {Foci}
\definetool{\tcp}       {TCP}
\definetool{\escjava}   {ESC/Java}
\definetool{\slam}      {SLAM}
\definetool{\yogi}      {Yogi}
\definetool{\jpf}       {JPF}
\definetool{\sycmc}     {SyCMC}
\definetool{\impact}    {Impact}
\definetool{\wolverine} {Wolverine}
\definetool{\ufo}       {UFO}
\definetool{\dctwo}       {DC2}
\definetool{\ultimate}  {\textsc{Ultimate Automizer}}
\definetool{\automizer} {Automizer}

\newcommand{\abstentities}{\abststates}%
\newcommand{\abstentity}{\abststate}%

\newcommand{\true}{\ensuremath{\mathit{true}}\xspace}

\newcommand{\sem}[1]{[\![ #1 ]\!]}
\newcommand{\abst}[1]{\ensuremath{\langle\!\langle #1 \rangle\!\rangle}\xspace}
\newcommand{\abstraction}{\abst{\cdot}}

\newcommand{\meet}{\sqcap}
\newcommand{\bigmeet}{\bigsqcap}

\newcommand{\cpa}{\ensuremath{\mathbb{D}}\xspace}

\newcommand{\lattice}{\mathcal{E}}

\newcommand{\Natsz}{\mathbb{N}_0}
\newcommand{\Bools}{\mathbb{B}}

\newcommand{\less}{\sqsubseteq}
\newcommand{\join}{\sqcup}
\newcommand{\bigjoin}{\bigsqcup}

\newcommand{\strengthenop}{\ensuremath{\mathord{\downarrow}}\xspace}
\newcommand{\strengthen}{\strengthenop}
\newcommand{\transconc}[1]{\smash{\stackrel{#1}{\rightarrow}}}
\newcommand{\transabs}[2]{\smash{\stackrel[#2]{#1}{\rightsquigarrow}}}

\newcommand{\mergesepop}{\mathsf{merge}^\mathsf{sep}}

\newcommand{\stopop}{\ensuremath{\mathsf{stop}}\xspace}
\newcommand{\stopsepop}{\mathsf{stop}^\mathsf{sep}}

\DeclareSymbolFont{bbsymbol}{U}{bbold}{m}{n}
\DeclareMathSymbol{\bbpi}{\mathbin}{bbsymbol}{"05}

\newcommand{\lookahead}{\ensuremath{\ell}\xspace}%
\newcommand{\atmtstates}{\ensuremath{Q}\xspace}%
\newcommand{\atmtrelation}{\ensuremath{\delta}\xspace}%

\newcommand{\transducerstate}{\ensuremath{\iota}\xspace}%
\newcommand{\transducerstates}{\ensuremath{J}\xspace}%

\newcommand{\absttransducerstates}{\ensuremath{\transducerstates}\xspace}%
\newcommand{\absttransducerstate}{\ensuremath{\transducerstate}\xspace}%
\newcommand{\final}{\ensuremath{F}\xspace}%
\newcommand{\initial}{\ensuremath{{\transducerstate_0}}\xspace}%

\newcommand{\atmtfinal}{\final}%
\newcommand{\atmtinit}{\initial}%

\newcommand{\outrelation}{\ensuremath{\lambda}\xspace}
\newcommand{\concernmap}{\ensuremath{\zeta}\xspace}

\newcommand{\transtrans}{\ensuremath{t}\xspace}
\newcommand{\outalphabet}{\ensuremath{\Theta}\xspace}

\newcommand{\outsymbols}{\outalphabet}

\newcommand{\outwords}{\ensuremath{{\outalphabet^\infty}}}

\newcommand{\inputword}{{\ensuremath{\bar{\sigma}}\xspace}}
\newcommand{\inputwordset}{{\ensuremath{\hat{\sigma}}\xspace}}
\newcommand{\concinputword}{\ensuremath{\bar{\sigma}}\xspace}

\newcommand{\precisions}{\ensuremath{\Pi}\xspace}
\newcommand{\precision}{\ensuremath{\pi}\xspace}

\newcommand{\abstdomain}{\ensuremath{D}\xspace}
\newcommand{\abstworddomain}{\ensuremath{\overline{D}}\xspace}
\newcommand{\inputdomain}{\ensuremath{{\abstworddomain_\inputcomp}}\xspace}
\newcommand{\outputdomain}{\ensuremath{{\abstworddomain_\outputcomp}}\xspace}

\newcommand{\precop}{\mathsf{prec}}

\newcommand{\artefacts}{\ensuremath{A}\xspace}

\newcommand{\artefact}{\ensuremath{a}\xspace}
\newcommand{\artifact}{\artefact}

\newcommand{\artifacts}{\artefacts}

\newcommand{\artefactjoin}{\ensuremath{{\sqcup}}\xspace}

\newcommand{\concat}{\circ}%

\newcommand{\outclosure}{\ensuremath{\mathsf{abstclosure}}\xspace}%
\newcommand{\epsclosure}{\ensuremath{\mathsf{epsclosure}}\xspace}%
\newcommand{\closuretermstates}{\ensuremath{{\mathsf{closureterm}}}\xspace}%

\renewcommand{\hat}{\widehat}

\newcommand{\langin}{\ensuremath{\mathcal{L}_\textit{in}}\xspace}
\newcommand{\langacc}{\ensuremath{\mathcal{L}_\textit{acc}}\xspace}

\newcommand{\langua}{\ensuremath{\mathcal{L}}\xspace}%
\newcommand{\abstpath}{\ensuremath{{\bar e}}\xspace}%
\renewcommand{\path}{\ensuremath{\abstpath}\xspace}

\newcommand{\atmtinactive}{{\qoff}}%
\newcommand{\symtransducer}{\ensuremath{\mathsf{T}}\xspace}%
\newcommand{\transducer}{\symtransducer}%
\newcommand{\transduction}{\ensuremath{\mathcal{T}}\xspace}%
\newcommand{\transductions}{\transduction}%
\newcommand{\transducers}{\ensuremath{\mathbb{T}}}%

\newcommand{\yarn}{\textsc{Yarn}\xspace}

\newcommand{\trans}[1]{\smash{\xrightarrow{#1}}}
\newcommand{\mergeop}{\ensuremath{\mathsf{merge}}\xspace}

\newcommand{\properties}{\ensuremath{\mathbb{S}}\xspace}

\newcommand{\targetop}{\mathsf{target}}

\renewcommand{\hat}{\widehat}

\newcommand{\formula}{\ensuremath{\vartheta}\xspace}

\newcommand{\cfatransitions}{\ensuremath{G}\xspace}
\newcommand{\controlflows}{\cfatransitions}

\newcommand{\abststates}{\ensuremath{E}\xspace}
\newcommand{\abststate}{\ensuremath{e}\xspace}

\newcommand{\concretestates}{\ensuremath{C}\xspace}

\newcommand{\concretes}{\concretestates}

\newcommand{\powset}[1]{\ensuremath{2^{#1}}\xspace}

\newcommand{\qoff}{{q_\pi}}
\newcommand{\qinactive}{\atmtinactive}
\newcommand{\qbot}{q_\bot}

\newcommand{\srccode}[1]{\texttt{#1}}

\newcommand{\aspects}{\ensuremath{H}\xspace}

\newcommand{\concerns}{\aspects}

\newcommand{\sementail}{\ensuremath{\vDash}\xspace}%

\newenvironment{component}[1]%
{
  \par
  \smallskip
  \noindent
  \textcolor{gray}{\mbox{\textsf{#1}.}}~%
  \noindent\ignorespaces
}%
{
  \par
  \smallskip
  \ignorespacesafterend
} 

{
  \par
  \smallskip
  \noindent
  \textcolor{gray}{\mbox{\textsf{#1}.}}~%
  \noindent\ignorespaces
}%
{
  \par
  \smallskip
  \ignorespacesafterend
} 

\newenvironment{operator}[1]%
{
  \par
  \smallskip
  \noindent
  \textcolor{gray}{\mbox{\textsf{\textsc{Operator} \textbf{#1}}.}}~%
  \noindent\ignorespaces
}%
{
  \par
  \smallskip
  \ignorespacesafterend
} 

{
  \par
  \medskip
  \noindent
  \textcolor{gray}{\mbox{\textsf{\textsc{Scenario} \textbf{#1}}.}}~%
  \noindent\ignorespaces
}%
{
  \par
  \medskip
  \ignorespacesafterend
} 

{
  \par
  \medskip
  \noindent
  \textcolor{gray}{\mbox{\textsf{\textsc{Case Study} \textbf{#1}}.}}~%
  \noindent\ignorespaces
}%
{
  \par
  \medskip
  \ignorespacesafterend
} 

\usepackage{wrapfig}

\usepackage{varwidth}
\definecolor{bancolor}{RGB}{62,96,111}

\tcbset{myboxstyle/.style={enhanced,skin=enhanced jigsaw,
attach boxed title to top left={xshift=-3mm,yshift=-\tcboxedtitleheight/2},
coltitle=black,varwidth boxed title=0.7\linewidth,
colbacktitle=#1!75, colback=white,sharp corners,
top=2ex,
boxed title style={sharp corners, boxrule=0pt},
underlay boxed title={
\fill[#1] (title.south west) -- (title.south-|frame.west)--++(0,-.5*\tcboxedtitleheight)--cycle;
\fill[#1] (title.north east) --++(-.5*\tcboxedtitleheight,0)--++(0,.5*\tcboxedtitleheight)--cycle;}},
    myboxstyle/.default=bancolor}

\newtcolorbox{mybox}[2][]{myboxstyle,
title={\hspace*{.5cm}#2\hspace*{.5cm}},#1}

\tikzset{
 ctrlstate/.style = {state,align=center,inner sep=2pt, minimum size=2mm},
 cfastate/.style = {ctrlstate},
 cfatargetstate/.style = {cfastate, double},
 compstate/.style = {cfastate, rounded rectangle, minimum height=5mm, minimum width=3em, inner sep=3pt},
 concept/.style = {cfastate, inner sep=3pt, fill=gray!50, rectangle, 
        minimum width=17mm, minimum height=9mm, draw=gray!6 },
 crosscutting/.style = {concept, rounded rectangle, fill=gray!30},
 conceptstate/.style = {concept, rounded rectangle, fill=gray!30, draw=gray!90, minimum width=20mm},
 inputstate/.style = {concept, rectangle, fill=gray!10, draw=gray!90, minimum width=20mm, inner sep=3pt},
 explain/.style = {circle, draw=gray!30, line width=1mm, minimum size=7mm},
 one/.style = {fill=blue!70,draw=blue!70},
 two/.style = {fill=red!50,draw=red!50},
 abststate/.style = {rectangle,align=center,inner sep=2pt,minimum size=3.5mm,fill=gray!0, draw=gray!90},
 line/.style = {draw},
 trans/.style = {draw,semithick,->,shorten >=1pt,>=stealth'},
 ctrans/.style = {draw,very thick,->,shorten >=1pt,>=stealth',draw=gray!90},
 epsilon/.style = {trans,dashed},
 strengthen/.style = {draw=gray!30,semithick,double,shorten >=1pt,>=stealth',line width=1mm},
}

\newcommand{\inputalphabet}{\ensuremath{\Sigma}\xspace}%
\newcommand{\outputalphabet}{\ensuremath{\Theta}\xspace}%
\newcommand{\inputcomp}{\ensuremath{{\text{in}}}\xspace}%
\newcommand{\outputcomp}{\ensuremath{{\text{out}}}\xspace}%
\newcommand{\abstractlang}{\ensuremath{\mathfrak{L}}\xspace}%
\newcommand{\abstractlangm}{\ensuremath{\mathfrak{M}}\xspace}%
\newcommand{\abstractword}{\ensuremath{\mathfrak{v}}\xspace}%
\newcommand{\abstractwordv}{\ensuremath{\mathfrak{w}}\xspace}%
\newcommand{\abstractwords}{\ensuremath{\mathfrak{I}}\xspace}%
\newcommand{\abstractwordsv}{\ensuremath{\mathfrak{W}}\xspace}%
\newcommand{\abstractepsilon}{\ensuremath{\abstractword_\epsilon}\xspace}%
\newcommand{\concretewords}{\ensuremath{W}\xspace}%
\newcommand{\concreteword}{\ensuremath{\bar{w}}\xspace}%
\newcommand{\inputwords}{\ensuremath{\abstractwords}\xspace}%
\newcommand{\outputwords}{\ensuremath{\abstractwordsv}\xspace}%
\newcommand{\outputlang}{\ensuremath{{\abstractlangm}}\xspace}%
\newcommand{\outputword}{\ensuremath{{\abstractwordv}}\xspace}%
\newcommand{\abstoutputword}{\outputword}%
\newcommand{\abstinputword}{\abstractword}%

\newcommand{\run}{\ensuremath{\bar{\transducerstate}}\xspace}%

\def\HiLi{\leavevmode\rlap{\hbox to \hsize{\color{yellow!50}\leaders\hrule height .8\baselineskip depth .5ex\hfill}}}

\begin{document}

\title{Abstract Transducers}
\titlenote{This is a preliminary report. Please refer to the final version of this work if available. }

\author{Andreas Stahlbauer}
\email{andreas@stahlbauer.net}
\orcid{0000-0003-4174-7242}
\affiliation{\institution{University of Passau}}

\begin{abstract}
Several abstract machines that operate on symbolic input alphabets 
have been proposed in the last decade, for example, symbolic 
automata or lattice automata.
Applications of these types of automata include
software security analysis and natural language processing.
%
While these models provide means to describe words over infinite 
input alphabets, there is no considerable work on symbolic output~(as 
present in transducers) alphabets, or even abstraction (widening) thereof. 
Furthermore, established approaches for transforming, for example,
minimizing or reducing, finite-state machines that produce output
on states or transitions are not applicable.
A notion of equivalence of this type of machines is needed
to make statements about whether or not transformations
maintain the semantics. 

We present \emph{abstract transducers} as a new form
of finite-state transducers.
Both their input alphabet and the output alphabet
is composed of \emph{abstract words}, where 
one abstract word represents a set of concrete words.
The mapping between these representations is described  
by \emph{abstract word domains}.
By using words instead of single letters, abstract transducers 
provide the possibility of lookaheads to decide on state transitions
to conduct.
Since both the input symbol and the output symbol on each 
transition is an abstract entity, \emph{abstraction} techniques 
can be applied naturally.

We apply abstract transducers as the foundation for 
\emph{sharing task artifacts for reuse} in context of program analysis
and verification, and describe task artifacts as abstract words.
A task artifact is any entity that contributes to an analysis 
task and its solution, for example, candidate invariants
or source code to weave.

\end{abstract}


\maketitle

\newpage
\section{Introduction}

\newcommand{\singlewordset}[1]{\ensuremath{\srccode{#1}}\xspace}%
\newcommand{\epsilontrigger}{\ensuremath{\{ \epsilon \}}\xspace}%

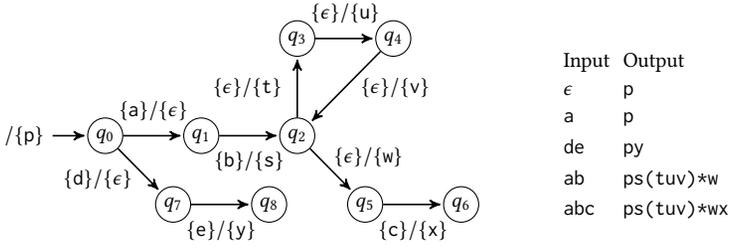
\begin{figure}[tp]
\centering
\begin{tikzpicture}[node distance=10mm, scale=0.8, transform shape,
cell/.style={rectangle,draw=black},
space/.style={minimum height=1.2em,matrix of nodes,row sep=-\pgflinewidth,column sep=-\pgflinewidth}]
\node[] (es) {};
\node[cfastate,right of=es] (e0) {$q_0$};
\node[cfastate,right=of e0] (e1) {$q_1$};
\node[cfastate,right=of e1] (e2) {$q_2$};
\node[cfastate,above=of e2] (e3) {$q_3$};
\node[cfastate,right=of e3] (e4) {$q_4$};
\node[cfastate,below right=of e0] (e7) {$q_7$};
\node[cfastate,right=of e7] (e8) {$q_8$};
\node[cfastate,below right=of e2] (e5) {$q_5$};
\node[cfastate,right=of e5] (e6) {$q_6$};
\path[trans] (es) edge node [pos=.2,label=left:{$/ \{ \singlewordset{p} \} $}] {} (e0);
\path[trans] (e0) edge node [pos=.5,label=above:{$\{ \singlewordset{a} \} / \{ \epsilon \}$}] {} (e1);
\path[trans] (e1) edge node [pos=.5,label=below:{$\{ \singlewordset{b} \} / \{ \singlewordset{s} \}$}] {} (e2);
\path[trans] (e2) edge node [pos=.5,label=left:{$\epsilontrigger / \{ \singlewordset{t} \} $}] {} (e3);
\path[trans] (e0) edge node [pos=.7,label=left:{$\{ \singlewordset{d} \} / \{ \epsilon \} $}] {} (e7);
\path[trans] (e7) edge node [pos=.5,label=below:{$\{ \singlewordset{e} \} / \{ \singlewordset{y} \} $}] {} (e8);
\path[trans] (e4) edge node [pos=.5,label=right:{$\epsilontrigger / \{ \singlewordset{v} \} $}] {} (e2);
\path[trans] (e3) edge node [pos=.5,label=above:{$\epsilontrigger / \{ \singlewordset{u} \} $}] {} (e4);
\path[trans] (e2) edge node [pos=.2,label=right:{$\epsilontrigger / \{ \singlewordset{w} \} $}] {} (e5);
\path[trans] (e5) edge node [pos=.5,label=below:{$\{ \singlewordset{c} \} / \{ \singlewordset{x} \}$}] {} (e6);
\matrix (first) at (10,0) [space, column 1/.style={anchor=base west, scale=0.8}, column 2/.style={anchor=base west, scale=0.8}]
{
 Input & Output \\ 
 $\epsilon$ & \srccode{p} \\
 \srccode{a} & \srccode{p} \\
 \srccode{de} & \srccode{py} \\
 \srccode{ab} & \srccode{ps(tuv)*w} \\
 \srccode{abc} & \srccode{ps(tuv)*wx} \\
};
%
%
\end{tikzpicture}
\caption{
Each transition of an abstract transducer is annotated with an 
abstract input word~\abstinputword and an abstract output word~\abstoutputword. 
An abstract input word denotes~$\sem{\abstinputword} \subseteq \inputalphabet^*$ 
a set of words over an input alphabet~\inputalphabet, 
an abstract output word denotes~$\sem{\abstoutputword} \subseteq \outputalphabet^\infty$ a set of words over an output alphabet~\outputalphabet; 
the illustration shows the sets of concrete words.
Please note that already the activation of the initial state~$q_0$
emits an abstract output word---in this case: the 
set~$\{ \srccode{p} \}$ of concrete words.
We use the abstract epsilon word~$\abstractepsilon$, 
with~$\sem{\abstractepsilon} = \{ \epsilon \}$,
with the semantics that is known from $\epsilon$-NFAs~\cite{HopcroftEtal2003,Sipser}. 
Transitions with the abstract input word~$\abstractepsilon$ are called
$\epsilon$-moves and can result in output words of infinite length, 
which result from $\epsilon$-loops---for example, the loop 
$q_2 \transconc{} q_3 \transconc{} q_4 \transconc{} q_2$.
The table on the right shows---for the transducer on the left---a 
set of input words and corresponding output words in the 
form of regular expressions.
}
\label{fig:artitrans:example}%
\label{fig:intro:ex1}%
\end{figure}

We present \emph{abstract transducers} as a new type of abstract machines 
that operate on an abstract input alphabet and an abstract output 
alphabet, and that have an inherent notion of abstraction.
Both the input alphabet and the output alphabet are described 
based on abstract domains, which enables different forms 
of abstracting these transducers and allows for different forms
of \emph{symbolic} representations.
An abstract representation of words is essential for creating 
\emph{finite abstractions} of possibly exponentially many and 
infinitely long output words, and abstraction of a transducer 
allows to increase the \emph{sharing} of its outputs, that is, 
one output becomes applicable to a wider set of input words.
Different abstract domains, and respective lattices, have been 
proposed to represent and abstract states and behaviors of systems
and their relationships~\cite{AbstractInterpretation}.
An abstract domain provides means to map between abstract and 
concrete entities.
Combining abstract domains and finite-state transducers 
results in a generic formalism that provides a \emph{unified view} 
on different types of automata and transducers, 
and enables \emph{new applications} in different areas, 
for example, in program analysis and verification.
Figure~\ref{fig:intro:ex1} illustrates the
working principle of abstract transducers.

\paragraph{Problems.}%
Abstract transducers address several problems:
(1)~In case alphabets consist of many, possibly exponentially many, 
symbols, traditional automata concepts with single concrete symbols 
per transitions provide limited efficiency.
Automata that employ a symbolic alphabet---where one symbol
from the alphabet denotes a set of concrete symbols---%
solve this issue~\cite{NoordGerdemann2001,Veanes2013}.
Having a symbolic representation of alphabet symbols makes approaches 
for abstracting~(or widening)~finite-state machines---such as 
relational abstraction or alphabet 
abstraction~\cite{BultanEtal2017,PredaEtal2015}---applicable.
We use abstract domains, as known from abstract interpretation,
for constructing symbolic representations, and mapping between
concrete and symbolic alphabets.
This way, we can choose from a large variety of abstract domains
to provide different symbolic and explicit mechanisms for representing 
data, for example, binary decision diagrams~\cite{BDDs}, 
predicates~\cite{PredicateAbstraction,BooleanCartesian}, 
or polyhedra~\cite{SinghEtal2017}.
%
Abstraction is also essential for output words, which are produced
by transducers, and has not yet received attention by researchers.
(2)~We allow the transducers to have $\epsilon$-moves that
are annotated with outputs, which can lead to output words
of \emph{infinite length}; here, a symbolic representation of sets 
of output words, based on corresponding abstract domains for the output 
alphabet, can help to provide a finite representation that represents 
or even overapproximates sets of exponentially many and infinitely long words.
%
By having a means for abstracting both the input alphabet and 
the output alphabet, we can implement further, more elaborated 
techniques with various applications.
We abstract our transducers to \emph{increase the sharing} of the 
output they emit. 
An abstract transducer might have been constructed to produce its 
output for a specific set of input words that can be found in a 
specific analysis task, that is, (3)~the reuse of the output can 
be limited to a specific set of analysis tasks, while the output 
would also be applicable to a broader set of tasks.
Sharing is increased if a given output word becomes produced
for a larger set of input words---that is, we take advantage of 
the nondeterminism that abstraction introduces~\cite{AvniKupferman2013}.
%
The alphabets from that these words can be composed of can~(in general) 
consist of arbitrarily complex entities~(symbols), 
for example, tuples of concrete letters as used for 
multi-track automata~\cite{BultanEtal2017}.
(4)~Nevertheless, also for these complex symbols, a means of 
abstraction is needed.
Constructing complex alphabets, and words thereof, based on abstract 
product domains~\cite{CortesiEtal2013} addresses this issue.

\begin{figure}[tp]
\begin{center}
\begin{subfigure}[b]{.49\linewidth}
\begin{tikzpicture}[node distance=10mm, scale=0.8, transform shape]
\node[] (qs) {};
\node[cfastate,below of=qs] (q0) {$q_0$};
\node[cfastate,below right=of q0] (q1) {$q_1$};
\node[cfatargetstate,below right=of q1] (q2) {$q_2$};
\path[trans] (qs) edge node [pos=.3,label=left:{$/\abstoutputword_\epsilon$}] {} (q0);
\path[trans] (q0) edge node [pos=.2,label=right:{$\abstinputword_1/\abstoutputword_\epsilon$}] {} (q1);
\draw[trans,bend left] (q1) edge node [pos=.5,label=left:{$\abstinputword_2/\abstoutputword_\epsilon$}] {} (q0);
\draw[trans] (q0) edge [out=10,in=50,looseness=8] node [pos=.5,label=right:{$\lnot \abstinputword_1/\abstoutputword_\epsilon$}] {} (q0);
\draw[trans] (q1) edge [out=10,in=50,looseness=8] node [pos=.5,label=right:{$\abstinputword_3/\abstoutputword_1$}] {} (q1);
\draw[trans] (q1) edge [out=190,in=230,looseness=8] node [pos=.5,label=left:{$\lnot \abstinputword_2/\abstoutputword_\epsilon$}] {} (q1);
\draw[trans] (q1) edge node [pos=.3,label=right:{$\abstinputword_3/\abstoutputword_2$}] {} (q2);
\node[] (t1) at (-0.5,-3.5) {$\sem{\abstoutputword_\epsilon} = \{ \epsilon \}$};
\node[below=0.1em of t1] (t2) {$\sem{\abstoutputword_1} = \{ \srccode{\$2 != 7} \}$};
\node[below=0.1em of t2] (t3) {$\sem{\abstoutputword_2} = \{ \srccode{\$2 == 7} \}$};
\node[] (t5) at (2.8,-4.7) {$\sem{\abstinputword_1} = \{ \srccode{enter()} \}$};
\node[below=0.1em of t5] (t6) {$\sem{\abstinputword_2} = \{ \srccode{leave()} \}$};
\node[below=0.1em of t6] (t7) {$\sem{\abstinputword_3} = \{ o \concat \srccode{alloc(\$1, \$2)} \; | \; o \in \outputalphabet \}$};
\end{tikzpicture}
\caption{\yarn Transducer}
\label{fig:ats:example:yarn}
\end{subfigure}
\begin{subfigure}[b]{.49\linewidth}
\begin{tikzpicture}[node distance=10mm, scale=0.8, transform shape]
\node[] (dummy) at (-0.5,0) {};
\node[] (qs) {};
\node[cfastate,below of=qs] (q0) {$q_0$};
\node[cfastate,below right=of q0] (q1) {$q_1$};
\node[cfastate,below right=of q1] (q2) {$q_2$};
\node[cfastate,below right=of q2] (q3) {$q_3$};
\path[trans] (qs) edge node [pos=.3,label=left:{$/\abstoutputword_\epsilon$}] {} (q0);
\path[trans] (q0) edge node [pos=.2,label=right:{$\abstinputword_1/\abstoutputword_1$}] {} (q1);
\draw[trans] (q0) edge [out=10,in=50,looseness=8] node [pos=.5,label=right:{$\lnot \abstinputword_1/\abstoutputword_\epsilon$}] {} (q0);
\draw[trans] (q1) edge [out=10,in=50,looseness=8] node [pos=.5,label=right:{$\lnot \abstinputword_1/\abstoutputword_1$}] {} (q1);
\draw[trans] (q1) edge node [pos=.2,label=right:{$\abstinputword_1/\abstoutputword_2$}] {} (q2);
\draw[trans] (q2) edge [out=10,in=50,looseness=8] node [pos=.5,label=right:{$\lnot \abstinputword_1/\abstoutputword_2$}] {} (q2);
\draw[trans] (q2) edge node [pos=.2,label=right:{$\abstinputword_1/\abstoutputword_3$}] {} (q3);
\draw[trans] (q3) edge [out=10,in=50,looseness=8] node [pos=.5,label=right:{$\top/\abstoutputword_3$}] {} (q3);
\node[] (t1) at (0.3,-3.2) {$\sem{\abstoutputword_\epsilon} = \{ \true \}$};
\node[below=0.1em of t1] (t2) {$\sem{\abstoutputword_1} = \{ i \equiv 1 \}$};
\node[below=0.1em of t2] (t3) {$\sem{\abstoutputword_2} = \{ i \equiv 2 \}$};
\node[below=0.1em of t3] (t4) {$\sem{\abstoutputword_3} = \{ i \equiv 3 \}$};
\node[] (t5) at (3.7,-5.2) {$\sem{\abstinputword_1} = \{ \srccode{i = i + 1} \}$};
\end{tikzpicture}
\caption{Precision Transducer}
\end{subfigure}
\end{center}
\caption{Two types of abstract transducers are illustrated: 
A \yarn transducer that emits code to weave, that is, it corresponds 
to an aspect in AOP, and a precision transducer, which is a means 
to share candidate invariants, or predicates, for (re-)use in 
predicate abstraction.
Please note that the abstract input word~$\abstinputword_3$ describes a 
lookahead, that is, it contains a word~$\bar{\sigma} \in \sem{\abstinputword}$ 
with~$|\bar{\sigma}| > 1$.
The lookahead matches if the input that remains to be consumed 
after word~$\abstinputword_3$ matched starts 
with~$\srccode{alloc(\$1,\$2)}$, where~$\srccode{\$1}$
and~$\srccode{\$2}$ are parameters that bind the arguments that
are given to the function~$\srccode{alloc}$ to internal variables,
which are then used to produce the concrete output for the 
abstract output words~$\abstoutputword_1$ and~$\abstoutputword_2$.}
\label{fig:example:ats}%
\label{fig:ex2}%
\end{figure}
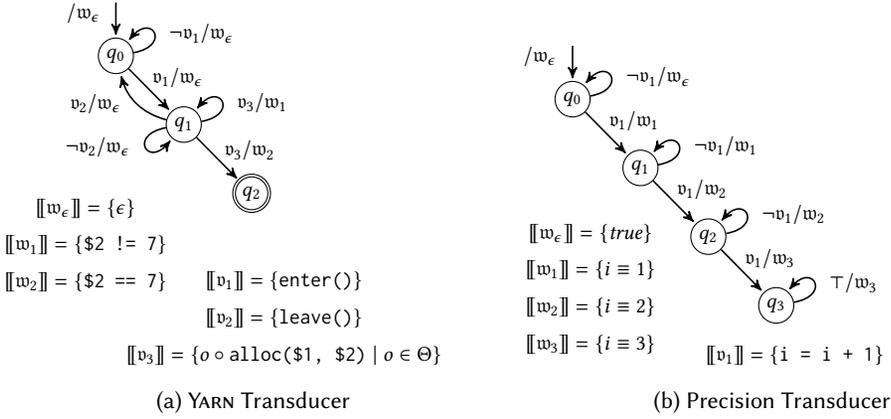

\paragraph{Applications.}%
We instantiate abstract transducers as task artifact transducers.
A \emph{task artifact transducer} is an abstract transducer that 
maps between a set of control paths of a given program to analyze 
and a set of task artifacts, which are intended to be shared for reuse.
Task artifact transducers are a \emph{generic means to provide 
artifacts that contributes to an analysis task and its solution}.
These task artifact transducers aid in various analysis tasks 
for that task artifacts, for example, intermediate verification 
results, have to be provided at specific points and in specific 
contexts in the control flow.
By using such transducers as means for sharing artifacts for reuse,
we \emph{gain precise control} over the sharing process: 
We can precisely specify at which points and in which context~(path 
prefix), of the control flow of a program, certain artifacts 
should be shared for reuse.
We use them both to construct the transition relation of the 
analysis task itself, and for constructing a state-space abstraction 
with a finite number of abstract states in an efficient and effective 
manner, that is, for sharing syntactic and semantic task artifacts. 
Syntactic task artifacts include, for example, components, aspects, or
assertions to check~\cite{AOP,SLIC}.
Semantic task artifacts include, for example, 
function summaries~\cite{SeryEtal2012}, invariants, 
or Craig interpolants~\cite{LazyAbstraction,AbstractionsFromProofs}.
The goal of sharing task artifacts is to make the overall 
process of constructing syntactic and semantic task models 
more efficient and effective.

We present two forms of task artifact transducers based on 
abstract transducers in another work~\cite{thesis:stahlbauer}:
\yarn transducers and precision transducers.
Figure~\ref{fig:example:ats} provides examples for these types
of abstract transducers.
A \emph{\yarn transducer} can express aspects---source code, 
or labeled transition systems (LTSs) in general, to emit at 
specific points---to weave into a control-flow graph. 
Such aspects can, for example, provide the environment 
model or a specification. It must be possible to emit code to 
weave before any of the transitions that are processed as input: 
An initial transducer output is needed. For soundness, operations 
such as $\epsilon$-elimination, union, or reduction must keep the 
semantics---including their temporal relationships, also 
concurrency---of these aspects.
A \emph{precision transducer} is annotated with sets of 
predicates~(candidate invariants) to emit for reuse in different 
contexts of the transition system to construct (for example, 
a Kripke structure) in an analysis process.
The shared predicates can be used to compute predicate abstractions~(as 
used for software model checking~\cite{PredicateAbstraction,BallEtal2001}), 
the number of CEGAR~\cite{CEGAR,PrecisionReuse} iterations can 
be reduced by abstracting these transducers, which increases 
sharing (the same predicate can be emitted in more contexts). 
Such precision transducers can also express the predicate
sharing strategy of lazy abstraction~\cite{LazyAbstraction}.

\subsubsection{Contributions}
This work presents the following contributions and shares
most of the material with the author's thesis~\cite{thesis:stahlbauer}:
\begin{itemize}
\setlength\itemsep{0.4em}
  \item \emph{Abstract Transducers.} 
    We introduce \emph{abstract transducers} as a generic 
    and unifying type of abstract machines that use \emph{abstract 
    word domains} to characterize both the input alphabet and the 
    output alphabet, and that have an inherent notion of \emph{abstraction}.

  \item \emph{Abstract Output Closure.}
    We present techniques for computing \emph{finite abstractions 
    of the output} of $\epsilon$-closures with \emph{$\epsilon$-loops},
    which are possible if $\epsilon$-moves are allowed.
    These techniques allow to produce finite outputs from transducers 
    with outputs that describe exponentially large sets of potentially 
    infinitely long words, and they aid in eliminating the $\epsilon$-moves.

  \item \emph{Transducer Abstraction.}
    We exactly define what it means to \emph{abstract} (or overapproximate)
    an abstract transducer. Based on this notion of abstraction
    we discuss different \emph{types of abstractions} and define 
    corresponding operators. 

  \item \emph{Transducer Reduction.}
    After defining the notion of \emph{equivalence} of abstract transducers,
    we discuss different transformations that maintain their semantics
    while \emph{reducing} their number of control states and transitions.
    Such reduction techniques help to reduce the degree of non-determinism,
    which reduces the costs of executing abstract transducers.
  
  \item \emph{Transducer Analysis.}
    We present an~\emph{abstract transducer analysis} as a 
    generic configurable program analysis for running different 
    types of abstract transducers.

  \item \emph{Task Artifact Transducers.}
    We instantiate abstract transducers as \emph{task artifact transducers} 
    to have a generic means to share various artifacts that contribute 
    to different concerns of an analysis task.
    Task artifact transducers foster the reuse of components of 
    an analysis task and the intermediate analysis~(reasoning) results 
    that are produced while conducting an analysis.

\end{itemize}

\subsubsection{Key Insights}
(1)~Abstract transducers have a fundamentally different semantic
compared to other transducers with symbolic alphabets, such as
symbolic transducers~\cite{DAntoniVeanes2017}, which becomes
obvious when comparing their notions of equivalence.
(2)~Existing algorithms for transforming finite-state transducers
are not applicable for abstract transducers.
(3)~Abstracting abstract transducers is a means for systematically
increasing the scope of sharing artifacts for reuse.

\section{Preliminaries}
\newcommand{\abstconcat}{\ensuremath{\varotimes}\xspace}%
\newcommand{\prefixof}{\ensuremath{\preceq}\xspace}%
\newcommand{\alphabet}{\ensuremath{\Sigma}\xspace}%
\newcommand{\words}{\ensuremath{\alphabet^*}\xspace}%
\newcommand{\infinitewords}{\ensuremath{\alphabet^{\omega}}\xspace}%
\newcommand{\eitherwords}{\ensuremath{\alphabet^{\square}}\xspace}%
\newcommand{\allwords}{\ensuremath{\alphabet^{\infty}}\xspace}%
\newcommand{\finitewords}{\ensuremath{\words}\xspace}%
\newcommand{\word}{\ensuremath{\bar{\sigma}}\xspace}%
\newcommand{\worda}{\ensuremath{\bar{a}}\xspace}%
\newcommand{\wordb}{\ensuremath{\bar{b}}\xspace}%
\newcommand{\wordc}{\ensuremath{\bar{c}}\xspace}%
\newcommand{\letter}{\ensuremath{\sigma}\xspace}%
\newcommand{\functions}{\ensuremath{F}\xspace}%
\newcommand{\function}{\ensuremath{f}\xspace}%
\newcommand{\functionof}{\ensuremath{\functions}\xspace}%
\newcommand{\powerlatticeof}[1]{\ensuremath{\mathsf{pw}(#1)}\xspace}%
\newcommand{\flatlatticeof}[1]{\ensuremath{\mathsf{fl}(#1)}\xspace}%
\newcommand{\maplatticeof}[2]{\ensuremath{\mathsf{ml}(#1, #2)}\xspace}%
\newcommand{\downpowerlatticeof}[1]{\ensuremath{\wp_\downarrow(#1)}\xspace}%
\newcommand{\downlatticeof}[1]{\downpowerlatticeof{#1}}%
\newcommand{\downsetsof}[1]{\ensuremath{\mathcal{I}(#1)}}%
\newcommand{\downsetof}[1]{\ensuremath{\downarrow #1}}%
\newcommand{\latticeof}[1]{\ensuremath{\dddot{#1}}}%
\newcommand{\latelems}{\ensuremath{E}\xspace}%
\newcommand{\keys}{\ensuremath{\textsf{keys}}}%
\newcommand{\imagejoin}{\ensuremath{\bigjoin_\rightarrow}\xspace}%
\newcommand{\prefixlatticeof}[1]{\ensuremath{\mathsf{pr}(#1)}\xspace}%
\newcommand{\concrete}{\ensuremath{C}\xspace}%
\newcommand{\semilattice}{\ensuremath{\mathcal{E}}\xspace}%

We start with preliminaries, including the notation:
We denote \emph{sets} $A,B,\ldots$ by upper case letters or add 
a hat~$\hat{a}, \hat{b}, \ldots$ to signal that an entity is a set.
\emph{Set elements}~$a, b, \ldots$ are denoted by lower case 
letters.
We add a bar~$\bar{a}, \bar{B}, \ldots$ to denote \emph{lists~(or sequences)}.
Elements of sets are enclosed~$\{ a, b, c \}$ in curly brackets, 
components of tuples are enclosed~$(A, x, y)$ in round brackets,
elements of lists are enclosed~$\langle a, b, c \rangle$ in angle brackets.

\subsubsection{Languages and Words}
The set of all \emph{finite words} over an \emph{alphabet}~$\Sigma$ 
is denoted by~$\Sigma^*$, which is a 
free monoid~$\Sigma^* = (\Sigma, \concat, \epsilon)$, 
also known as Kleene star, where 
\emph{concatenation}~$\concat: \words \times \words \rightarrow \words$
is its binary operator and its neutral element~$\epsilon$ is 
the \emph{empty word}.
A \emph{word}~$\bar{\sigma}$ is a sequence~$\langle \sigma_1, \ldots, \sigma_n \rangle$ 
of symbols from the alphabet.
The \emph{length}~$|\bar{\sigma}|$ of a word~$\bar{\sigma}$
is its number of subsequent symbols; the empty word 
has length~$|\epsilon| = |\langle \rangle| = 0$.
Given two words~$\bar{\sigma} = \langle \sigma_1, \ldots, \sigma_n \rangle$ 
and $\bar{\tau} = \langle \tau_1, \ldots, \tau_m \rangle$,
the concatenation~$\bar{\sigma} \concat \bar{\tau}$ 
results in the word~$\bar{c} = \langle \sigma_1, \ldots, \sigma_n, 
\tau_1, \ldots, \tau_m \rangle$ 
with length~$|\bar{c}| = |\bar{\sigma}| + |\bar{\tau}|$.
A word~$\word_a$ is \emph{prefix} of another word~$\word_b$, that is,
it is element~$(\word_a, \word_b) \in \prefixof$ of the prefix
relation~$\prefixof$, if there exists a \emph{suffix}~$\word_s \in \Sigma^*$ 
such that~$\word_b = \word_a \concat \word_s$.
While the finitely words over an alphabet~$\alphabet$ are denoted 
by~$\finitewords$, the \emph{infinitely long} ones are denoted 
by~$\infinitewords$, and the set of all words is denoted 
by~$\allwords = \finitewords \cup \infinitewords$~\cite{LodingTollkoetter2016},
with the \emph{infinite iteration}~$\cdot^\omega$ and 
the \emph{finite iteration}~$\cdot^*$.
The set of all words over an alphabet~$\Sigma$ that is 
described by a structure~$S$ and that are considered well-formed 
regarding certain production rules is called 
\emph{language}~$\langua(S) \subseteq \Sigma^*$ of $S$.
The \emph{empty word language} $\{ \epsilon \}$ consists of 
the empty word~$\epsilon$ only, the \emph{empty language} 
corresponds to the empty set~$\emptyset$.

\subsubsection{Lattices}
A (complete) \emph{lattice}~\cite{Graetzer2011} is a 
tuple~$\latticeof{E} = (\latelems, \less, \meet, \join, \top, \bot)$, 
with a set of \emph{abstract elements}~$\latelems$ and a 
\emph{partial order relation}~$\less \; \subseteq \latelems \times \latelems$.
The operator~\emph{meet} is a relation~$\meet: (\latelems \times \latelems) \rightarrow \latelems$ 
that provides the \emph{greatest lower bound} for a given 
pair~$(\abstentity_1, \abstentity_2) \in \latelems \times \latelems$
of abstract elements.
The operator \emph{join} is a relation~$\join: \latelems \times \latelems \rightarrow \latelems$ 
that provides the \emph{least upper bound} for a given pair 
of abstract elements.
The \emph{bottom} element~$\bot$ is the \emph{least} element in 
the partial order relation, that is, there exists no other 
element~$\abstentity \less \bot$, with~$\abstentity \neq \bot$.
The \emph{top} element~$\top$ is the \emph{greatest} element in 
the partial order, that is, there exists no other 
element~$\top \less \abstentity$, with~$\abstentity \neq \top$.
The operators~$\meet$ and~$\join$ extend to sets naturally: 
The meet over a set of abstract elements is denoted 
by~$\bigmeet: 2^\latelems \times \latelems$, and the join 
by~$\bigjoin: 2^\latelems \times \latelems$, for example, 
$\bot = \bigmeet \latelems$ and $\top = \bigjoin \latelems$.
%
A \emph{semi lattice} either has no meet or no join 
for all abstract elements.
The relation~$\less$ is also called the \emph{inclusion relation}~\cite{Whitman1946}.
Partial ordered sets (\emph{posets}) can be made semi-lattices, and semi-lattices 
can be made \emph{complete} by adding additional abstract elements~\cite{LatticeTheoryComputerScience}.
A lattice element~$e_c \in \latelems$ is called \emph{complement}
of an element~$e \in \latelems$ iff~$e \meet e_c = \bot$
and~$e \join e_c = \top$. 
We denote the complement of an lattice element~$e$ by~$\lnot e$.
A lattice is called~\emph{complemented} iff there exists
a complement for all its elements.

\subsubsection{Powerset Lattice}
The powerset lattice that describes a Hoare powertheory~\cite{HandbookOfLogic,FergusonHughes1989,HartTsinakis2007}---%
over a given lattice~$\latticeof{E} = (\latelems, \less_E, \meet_E, \join_E, \top_E, \bot_E)$---%
is denoted by~$\powerlatticeof{\latticeof{E}} = (2^\latelems, \less, \meet, \join, \top, \bot)$,
where the set of elements is constituted by the set of all 
subsets~$2^\latelems$ of set~$\latelems$.
The inclusion relation~$\less$ has the element~$(E_1, E_2) \in \less$ 
if and only if~$\forall e_1 \in E_1 \exists e_2 \in E_2: e_1 \less_E e_2$.
The join~$\join(E_1, E_2) = E_1 \cup E_2$ is the union,
and the meet~$\meet(E_1, E_2) = E_1 \cap E_2$ is the intersection
of two given sets~$E_1, E_2 \subseteq E$.
The bottom element~$\bot = \emptyset$ is the empty set,
and the top element~$\top = E$ corresponds to the set with all elements.

Any complemented distributive lattice is isomorphic to a Boolean 
algebra~\cite{Huntington1904},
which also follows from the Stone duality~\cite{Stone1936};
one example for such lattices are powerset lattices.
Lattices generalize Boolean algebras by not requiring 
complement and distributivity in the first hand.

\subsubsection{Map Lattice}
A \emph{map lattice}~$\maplatticeof{K}{\latticeof{V}} = (2^{K \rightarrow V}, \less, \meet, \join, \top, \bot)$
is a lattice of elements that are maps, that is, the elements are
functions that map from a set~$K$ of keys to a set~$V$ of values; 
the values of this map are elements of another  
lattice~$\latticeof{V} = (V, \less_V, \meet_V, \join_V, \top_V, \bot_V)$.
Such a lattice are also known as \emph{function lattice}~\cite{DuffusJonssonRival1978,BackWright1990}.
The inclusion relation~$\less$ has element~$(m_1, m_2) \in \less$ 
if and only if $\forall k \in K:  m_1(k)_\bot \less_V m_2(k)_\bot$.
In the following, we rely on the function $m(k)_\bot = m(k) \; \text{if} \; (k, \cdot) \in m \; \text{otherwise} \; \bot_V$ 
which returns the value for a given key~$k$ from a map~$m$, and the 
bottom element of the value lattice if no entry for the key is present.
The meet~$\meet$ is defined by~$\meet(m_1, m_2) = \{ (k, v_1 \meet_V v_2) \, | \, (k, v_1) \in m_1 \land v_2 = m_2(k)_\bot \}$,
the join~$\join$ is defined by $\join(m_1, m_2) = \{ (k, m_1(k)_\bot \join_V m_2(k)_\bot) \, | \, k \in \keys(m_1) \cup \keys(m_2) \}$,
the top element~$\top$ is defined by~$\top = \{ (k, \top_V) \, | \, k \in K \}$,
and the bottom element~$\bot$ is defined by~$\bot = \{ (k, \bot_V) \, | \, k \in K \}$.

We define an \emph{image-join operator}~$\imagejoin: 2^{K \times E} \rightarrow 2^{K \rightarrow E}$:
Given a map~$M \subseteq K \times E$, with a set of keys~$K$, and 
a set of lattice elements~$E$, the operator joins all tuples~$(k, e) \in M$ 
with the same key~$k$ into one tuple with a value that aggregates 
all value elements~$e$, that is, 
$\imagejoin M = \{ \; (k, \bigjoin \, \{ e \; | \; e \in \{ (k, e) \in M \} \} ) 
    \; | \; k \in \{ k \; | \; (k, \cdot) \in M \} \; \} 
$.

\subsubsection{Abstract Domains}
\newcommand{\powerdomainof}[1]{\ensuremath{\wp(#1)}\xspace}%
\newcommand{\downpowerdomainof}[1]{\ensuremath{\wp_\downarrow(#1)}}%
\newcommand{\powerdomain}{\ensuremath{\hat{\abstdomain}}}%
An abstract domain~$\abstdomain = (\latticeof{\concretes}, \latticeof{E}, \sem{\cdot}, \abst{\cdot})$~\cite{CousotCousot1992} 
is defined based on a tuple that consists of a lattice~\latticeof{\concretes} 
of the set of concrete elements~$\concretes$, a lattice~$\latticeof{E}$ 
of the set of abstract elements~$\abstentities$,
a denotation function~$\sem{\cdot}$ and an abstraction function~$\abst{\cdot}$.
The set $\concretes$ consists of all possible interpretations of 
elements from the set of abstract elements~$\abstentities$ for 
a specific universe.
The \emph{denotation}~$\sem{\abstentity}: \abstentities \rightarrow \powset{\concretes}$ 
of an abstract element~$\abstentity$ is the set of all its 
possible interpretations---as known from denotational semantics~\cite{HandbookOfLogic}.
The \emph{abstraction}~$\abst{\abstentity}$ of an abstract 
element~\abstentity results in a new abstract element~$e'$ 
with~$\sem{e} \subseteq \sem{e'}$. 
The abstraction~$\abst{\concretes_k}$ of a set of concrete 
elements~$\concretes_k \subseteq \concretes$ results in 
an abstract element~$\abstentity$, with~$\sem{\abstentity} = \concretes_k$. 
An abstraction~$\abst{\cdot}^\precision$ with \emph{widening} produces an 
abstraction with an abstraction precision~$\precision$, which 
can result in a widening.
The \emph{abstraction precision}~$\precision \in \precisions$~\cite{NayakLevy1995} 
defines the set of details that the resulting abstraction 
should maintain for sound reasoning.
Two elements are called \emph{semantically equal}, that is, 
$\abstentity_1 \equiv \abstentity_2$, if and only 
if~$\sem{\abstentity_1} = \sem{\abstentity_2}$ in the same universe.
One element \emph{semantically implies} another element,
that is, $\abstentity_1 \sementail \abstentity_2$, if and only
if~$\sem{\abstentity_1} \subseteq \sem{\abstentity_2}$.


\section{Abstract Words}
\label{sec:at:abstwords}%
\newcommand{\wordlattice}{\ensuremath{\dddot{\mathcal{W}}}\xspace}%
\newcommand{\finabstractwords}{\ensuremath{{\abstractalphabet^*}}\xspace}%
\newcommand{\abstractwordlattice}{\ensuremath{\dddot{\abstractalphabet}}\xspace}%
\newcommand{\concretealphabet}{\ensuremath{\Sigma}\xspace}%
\newcommand{\finconcretewords}{\ensuremath{{\concretealphabet^*}}\xspace}%
\newcommand{\orbottom}[1]{\ensuremath{\mbox{or}_{\bot}(#1)}\xspace}%
\newcommand{\conccomp}[1]{\ensuremath{#1}\xspace}%
\newcommand{\concbot}{\conccomp{\bot}}%
\newcommand{\conctop}{\conccomp{\top}}%
\newcommand{\concless}{\conccomp{\less}}%
\newcommand{\concjoin}{\conccomp{\join}}%
\newcommand{\concmeet}{\conccomp{\meet}}%
\newcommand{\abstcomp}[1]{\ensuremath{#1^\#}\xspace}%
\newcommand{\abstbot}{\abstcomp{\bot}}%
\newcommand{\absttop}{\abstcomp{\top}}%
\newcommand{\abstless}{\abstcomp{\less}}%
\newcommand{\abstjoin}{\abstcomp{\join}}%
\newcommand{\abstmeet}{\abstcomp{\meet}}%

Before we present abstract transducers, we describe concepts to cope 
with sets of possibly exponentially many and infinitely long words symbolically.
A word can \emph{express temporal or causal relationships} between 
the letters of the word.
We introduce concepts and techniques to deal with sets of 
words on an abstract level.

\subsection{Hierarchy of Characters, Words, and Languages}
\label{sec:chars:words:langs}%
We now discuss established terms that are relevant in the 
context of the terms that we introduce in the following sections.
This helps to understand our terminology choices.

Both the input alphabet and the output alphabet of an abstract
transducer is characterized based on an abstract domain.
\emph{Abstract domains} are a generic means for abstraction and 
provide various operations for manipulating and comparing abstract 
elements~(entities)~\cite{AbstractInterpretation}, and for mapping 
between \emph{concrete and abstract elements}.

Elements from a set~$\Sigma$ can be combined to form~(possibly infinite) 
sequences~$\bar{\sigma} \in \Sigma^\infty$ of those elements.
We use the term \emph{word} to denote sequences of elements that can 
be formed from other words by concatenation.
Words are elements of a free monoid~(semigroup) for that concatenation
is the binary and associative operator, and the empty word~(empty 
sequence) is the identity~(neutral) element.
A \emph{language} is a set of words---and typically well-formed
regarding some production rules.

In a generic abstract domain, one \emph{abstract element} maps to 
a set of \emph{concrete elements}, which is reflected by the 
denotation~(concretization) function~$\sem{\cdot}$.
That is, we can deduce that one \emph{abstract word} represents 
a set of \emph{concrete words}, and an \emph{abstract language} 
maps to a set of \emph{concrete languages}.

A word, as mentioned earlier, establishes a \emph{temporal
relationship} between all its characters; each character has 
a semantic denotation on its own, that is, it
maps to a set of entities.
The expressiveness of words compared their characters is dual
to the expressiveness of linear temporal logic to propositional logic:
A formula in propositional logic~(interpreted for a specific universe) 
denotes a set of entities, whereas a formula in linear temporal 
logic denotes sequences of sets of entities~(over time).
A \emph{set of words}, that is, a \emph{language}, provides sufficient
expressiveness to describe a set of forks in words over time, 
for example, to describe a set of \emph{concurrent} program executions,
or for matching trees or (more general) graphs.

That is, an abstract word, which maps to a set of concrete words, 
provides an abstraction with sufficient expressiveness to describe sets 
of linear-time concerns,
and an abstract language, which represents a set of sets of words,
provides expressiveness to describe sets of concerns that are 
expressible in branching-time logic. 
In the following, we restrict the discussion and presentation
\important
to abstract words and keep abstract languages for future work.

\subsection{Abstract Word Domain}
\label{concept:abstworddomain}%
The foundation of abstract transducers is formed by the 
abstract word domain, a lattice-based abstract 
domain~\cite{AbstractInterpretation,FileEtal1996} 
for mapping between abstract words and concrete words.
\newcommand{\epsilonlang}{\ensuremath{\abstractlang_\epsilon}\xspace}%
\newcommand{\epsilonword}{\ensuremath{\abstractword_\epsilon}\xspace}%
\newcommand{\emptywordlang}{\ensuremath{\epsilonlang}\xspace}%
\newcommand{\concretewordlattice}{\ensuremath{\latticeof{W}}\xspace}%
\newcommand{\wordlatticeof}[1]{\ensuremath{\omega(#1)}\xspace}%
\newcommand{\alphabetlattice}{\ensuremath{\latticeof{\alphabet}}\xspace}%

\begin{definition}[Abstract Word]
An \emph{abstract word}~$\abstractword \in \abstractwords$ 
is a symbolic representation of a set~$\subseteq \concretealphabet^\infty$ 
of concrete words over a concrete alphabet~$\concretealphabet$, 
where the set~$\abstractwords$ denotes all abstract words.
\end{definition}

\noindent
The relationship between an abstract word and the set of 
concrete words it represents, along with a means for abstraction,
is defined by the abstract word domain:
\nopagebreak[4]
\begin{definition}[Abstract Word Domain]
An \emph{abstract word domain} is an abstract 
domain~$\abstdomain_\concretewords = (\powerlatticeof{\latticeof{\concretewords}},
\latticeof{\abstractwords},\allowbreak \sem{\cdot}, \allowbreak \abst{\cdot})$
that has abstract words~\abstractwords as its abstract elements.
The relationship between abstract words is defined based on the
\emph{abstract word lattice}~$\latticeof{\abstractwords} = 
(\abstractwords, \less, \meet, \join, \top, \bot)$.
One abstract word~$\abstractword$ maps to a set of concrete
words~$\sem{\abstractword} \subseteq \concretewords$, 
which is defined by the denotation 
function~$\sem{\cdot}: \abstractwords \rightarrow 2^{\concretewords}$.
The lattice of concrete words~$\latticeof{\concretewords}$
defines the relationship between elements from
the set of concrete words~$\concretewords$.
Sets of concrete words are formed based on a powerset 
lattice~$\powerlatticeof{\latticeof{\concretewords}}$.
The abstraction function~$\abst{\cdot}: 2^\concretewords \rightarrow \abstractwords$
transforms a given set of concrete words~$\hat{\concreteword} \subseteq \concretewords$
into an abstract word~\abstractword, 
that is, $\abstractword = \abst{\hat{\concreteword}}$.
The \emph{abstract epsilon word}~\abstractepsilon maps~%
$\sem{\abstractepsilon} = \{ \epsilon \}$ to the set
with the empty word~$\epsilon$ only.
The \emph{bottom element}~$\bot$, or also \emph{abstract bottom word},
of the abstract word lattice denotes an abstract word that maps to 
the empty set of concrete words, that is, $\sem{\bot} = \emptyset$.
\end{definition}

\noindent
The abstraction mechanism that is provided by the abstract word 
domain is important for (1)~constructing finite 
abstractions of collections with exponentially many or infinitely 
long words; it can be used to (2)~check whether or not the 
analysis process ran into a fixed point, and (3)~for increasing 
the sharing of the output that we produce based on abstract transducers.

A problem that we have to deal with is the \emph{word coverage 
problem}, that is, the question of whether or not a given abstract 
word~$\abstractword_a$ is covered by another abstract 
word~$\abstractlang_b$, that is, if~$\abstractword_a \less \abstractword_b$, 
where~$\less$ is the inclusion relation of the abstract 
word lattice.
The actual matching process, that is, the check for coverage
can be implemented based on quotienting: The abstract word domain 
must provide the possibility to compute left
quotients~\cite{Brzozowski1964}~(Brzowzowski derivates) to 
match abstract words.

\newcommand{\quotientof}[2]{\ensuremath{#1^{#2}}\xspace}%
\begin{definition}[Left Quotient]The \emph{left quotient}~\cite{Brzozowski1964}~$\quotientof{\abstractword}{\abstractwordv}: 
\abstractwords \times \abstractwords \rightarrow \abstractwords$
of an abstract word~$\abstractword \in \abstractwords$ 
regarding an abstract word~$\abstractwordv \in \abstractwords$ is defined as 
$\quotientof{\abstractword}{\abstractwordv} = \abst{ \{ \bar{s} \; | \;
\bar{p} \concat \bar{s} \in \sem{\abstractword} \land \bar{p} \in \sem{\abstractwordv} \} }$.
It denotes suffixes of~$\abstractword$ for 
that~$\abstractwordv$ contains prefixes.
\end{definition}

\noindent
Another fundamental operation when dealing with words is their 
concatenation, which is the binary operator of the free 
monoid~$\alphabet^*$ that describes the set of words over an 
alphabet~$\alphabet$.
We extend this operator to abstract words, and with it
to sets of words:
\nopagebreak[4]
\begin{definition}[Concatenation]The \emph{concatenation} of a pair of abstract 
words~$\abstractword_1 \concat \abstractword_2$ results in 
an abstract word~$\abstractword_\concat$ that 
denotes~$\sem{\abstractword_\concat}$ the concatenation of all 
concrete finite words from the abstract word~$\abstractword_1$ with all 
(finite and infinite) concrete words from the abstract word~$\abstractword_2$.
The concatenation~$\word_1 \concat \word_2$ of an infinite word~$\word_1$
with another word~$\word_2$ results in the infinite word~$\word_1$.
That is~$\sem{\abstractword_1 \concat \abstractword_2} = \{ \word_1 \concat \word_2 \; | \; \word_1 \in \sem{\abstractword_1} \land \word_2 \in \sem{\abstractword_2} \}$.
\end{definition}

\noindent
To deal with abstract words, the notion of head and tail is
important:
\nopagebreak[4]
\newcommand{\headof}[1]{\ensuremath{\textsf{head}(#1)}\xspace}%
\begin{definition}[Head]Given an abstract word~$\abstractword$, the 
function~$\headof{\abstractword}: \abstractwords \rightarrow \abstractwords$ 
denotes the~\emph{head of an abstract word}:
The resulting abstract word represents the set of prefixes 
with length one, or formally~$\sem{\headof{\abstractword}} 
= \{ \, \bar{h} \; | \; \bar{h} \concat \word \in \sem{\abstractword} 
    \land |\bar{h}| = 1 \, \}$.
\end{definition}

\newcommand{\tailof}[1]{\ensuremath{\textsf{tail}(#1)}\xspace}%
\begin{definition}[Tail]The \emph{tail of an abstract word} is provided by the 
function~$\tailof{\abstractword}: \abstractwords \rightarrow \abstractwords$.
A call~$\abstractword' = \tailof{\abstractlang}$ returns a new 
abstract word~$\abstractword'$ that represents the set of 
postfixes that follow after the head. 
That is, $\sem{\tailof{\abstractword}} 
= \{ \, \word \; | \; \bar{h} \concat \word \in \sem{\abstractword} \land |\bar{h}| = 1 \, \}$, 
which equals~$\tailof{\abstractword} = \quotientof{\abstractword}{\headof{\abstractword}}$.
\end{definition}


\subsection{Boolean Algebra}
In several occasions, when reasoning about abstract words and 
their relationship, we need the full expressive power of a Boolean algebra.
We can build on the duality between Boolean algebras, regular 
languages, and complemented and distributive lattices,
which follows from the Stone duality~\cite{Stone1936,Pippenger1997}.
The abstract word lattice is dual to a Boolean algebra if 
and only if its meet~$\meet$ and join~$\join$ are distributive 
over each other and if each element in the lattice has a 
complement within the lattice.
One example of a lattice that is dual to a Boolean algebra is 
the powerset lattice and another one the
lattice of regular languages~\cite{GehrkeEtal2008,BrancoPin2009}.
Both lattices can describe sets of words and can thus be 
instantiated as an abstract word lattice of an abstract word domain. 
Given a lattice of regular expressions,
the join~$\join$ corresponds to the language inclusion,
the meet~$\meet$ to the language intersection,
and the operator~$\less$ describes the language inclusion;
the language is complemented since the complement of a 
regular expression is still regular.
\begin{definition}[Abstract Word Complement]Given an abstract word~$\abstractword$, 
its complement~$\lnot \abstractword$ defines a set of concrete words such 
that~$\forall \abstractword \in \abstractwords: \lnot \abstractword 
\meet \abstractword = \bot$
and~$\forall \abstractword \in \abstractwords: \lnot \abstractword 
\join \abstractword = \top$, 
with~$\forall \worda \in \sem{\abstractword}: \forall \wordb \in 
\sem{\lnot \abstractword}: \worda \meet_c \wordb = \bot_c$,
where~$\meet_c$ and $\bot_c$ are components of the concrete word lattice.
\end{definition}

\noindent
In case an abstract word lattice is dual to a Boolean algebra,
the abstract words and their composition, can also be described
using Boolean operators, which have their duals
in lattice theory: The join~$\join$ corresponds to the logical
disjunction~$\lor$, the meet~$\meet$ corresponds to the 
disjunction~$\lor$, and the complement corresponds to the logical negation~$\lnot$.
A Boolean formula~$\formula$ is equivalent to an abstract word~$\abstractword$
if and only if~$\sem{\formula} = \sem{\abstractword}$.


%

\subsection{Parameterized Words}
\label{concept:paramlang}%
\newcommand{\parameters}{\ensuremath{\mathcal{B}}\xspace}%
\newcommand{\values}{\ensuremath{\mathcal{V}}\xspace}%
\newcommand{\configarg}{\ensuremath{\mathcal{v}}\xspace}%
\newcommand{\arguments}{\ensuremath{\values}\xspace}%
\newcommand{\argument}{\ensuremath{\value}\xspace}%
\newcommand{\parameter}{\ensuremath{\mathbb{b}}\xspace}%
\newcommand{\bind}{\ensuremath{\textsf{bind}}\xspace}%
\newcommand{\bounded}{\ensuremath{\textsf{bounded}}\xspace}%
An abstract word can be \emph{parameterized} with a finite 
set of parameters~$\beta \subseteq \parameters$.
A parameterized abstract word can take two roles:
It~(1) can \emph{capture} (bind) values to the parameters during 
a matching process for a given input, and~(2) values for the
arguments can get passed explicitly (and act as a \emph{template}).
We use the term~\emph{instantiation} to denote the process of 
deriving an abstract word~$\abstractword'$
from an abstract word~$\abstractword$ by assigning values to the 
parameters, with~$\sem{\abstractword'} \subset \sem{\abstractword}$.
Examples for different types of templates words include
invariant templates~\cite{SrivastavaGulwani2009,KongEtal2010}.
The values that have been bound to the parameters of
an abstract word are provided by the operator~$\bounded: 
\abstractwords \rightarrow 2^{\parameters \rightarrow \values}$.
We can bind values to parameters of an abstract word and derive
a new abstract word with the operator~$\bind: \abstractwords \times 2^{\parameters \rightarrow \values} \rightarrow \abstractwords$.
Binding of values to parameters~(variable binding) was extensively
studied in the past, for example, for rewriting systems~\cite{NipkowPrehofer1998,Hamana2003}, and regular expressions~\cite{Freydenberger2013}.

\section{Abstract Transducers}
\label{sec:spat}%
\label{sec:analysistransducers}%
\label{sec:artefacttransdcuers}%
\label{sec:absttransducers}%
\label{concept:absttransducers}%
\newcommand{\coninalphabet}{\ensuremath{\Sigma}\xspace}%
\newcommand{\conoutalphabet}{\ensuremath{\outsymbols}\xspace}%
\newcommand{\inputartifacts}{\ensuremath{\artifacts_\coninalphabet}\xspace}%
\newcommand{\inputartifact}{\ensuremath{\artifact}\xspace}%
\newcommand{\outputartifacts}{\ensuremath{\artifacts_\conoutalphabet}\xspace}%
\newcommand{\outputartifact}{\ensuremath{\artifactb}\xspace}%
\newcommand{\abstinputwords}{\abstractwords}%
\newcommand{\abstoutputwords}{\outputwords}%
\newcommand{\transducerss}{\ensuremath{Z}\xspace}%
\newcommand{\inputlangeps}{\ensuremath{{\abstractlang_\epsilon}}\xspace}%
\newcommand{\outputlangeps}{\ensuremath{{\outputlang_\epsilon}}\xspace}%
\newcommand{\epsilonoutlang}{\ensuremath{\outputlangeps}\xspace}%
\newcommand{\emptylang}{\ensuremath{{\bot}}\xspace}%
\newcommand{\setrelation}{\ensuremath{\hat\atmtrelation}\xspace}%
\newcommand{\setrelationabst}{\ensuremath{\hat\atmtrelation_{\#}}\xspace}%
This work introduces abstract transducers, a type of abstract machines 
that map between abstract input words and abstract output words.
Compared to established transducer concepts, intermediate languages 
are central (we still have a  notion of accepted language):
Informally speaking, the \emph{intermediate input language} is 
the set of words for which the transducer can perform state 
transitions, and the set of words that are produced as output 
along these transitions is called the \emph{intermediate output language}.

To produce the intermediate output language, an abstract transducer
operates~\emph{prescient}, that is, it can take a lookahead 
into account to decide whether to conduct a state transition or 
not---and with it produce an output.
Words from the intermediate output language are intended
to be used immediately, that is, as soon as they are produced
while executing the transducer, which has several implications
on the design on the algorithms that execute abstract transducers
and that manipulate them---for example, to eliminate $\epsilon$-moves.

Both the input alphabet and their output alphabet are abstract
and defined based on abstract word domains.
One abstract word maps to a set of concrete words; the
abstract domain provides means for mapping between these representations.
This abstraction functionality enriches the possibilities to compute 
abstractions~(widenings) of abstract transducers, which 
we use as a \emph{means of increasing the scope of sharing}: one 
output is mapped to a larger set of inputs.

Each transition of an abstract transducer is annotated with an 
abstract input word and an abstract output word---which 
corresponds to symbols of the input alphabet and the output 
alphabet of traditional transducers.
Consuming and producing abstract words instead of single concrete 
letters has several advantages that increase the generality of 
our approach: (1)~it can be used for lookahead-matching, that is, 
instead of describing the input symbol to consume, also a sequence of 
symbols that must follow can be described, (2)~the abstract epsilon 
word~$\abstractepsilon$, with~$\sem{\abstractepsilon} = \{ \epsilon \}$,
can be used to model the behavior of 
an $\epsilon$-NFA~\cite{Sipser} with a corresponding
$\epsilon$-closure and to model automata that do not produce 
outputs at all, and (3)~relying on abstract words allows to 
produce and cope with output words of infinite length, which 
can be the result of $\epsilon$-loops.

\medskip
Formally, we define an abstract transducer as: 
\label{concept:atmtrelation}
\begin{definition}[Abstract Transducer]An \emph{abstract transducer}~$\symtransducer \in \transducers$ 
is defined by following tuple:
$$\transducer =  (\atmtstates, \inputdomain, \outputdomain, 
        \atmtinit, \atmtfinal, \atmtrelation)$$
\begin{itemize}
\setlength\itemsep{0.4em}
\item \emph{Control States~\atmtstates.} 
    The finite set~$\atmtstates$ defines the \emph{control states} 
    in which the transducer can be in.
    \label{concept:atmtstates}%

\item \emph{Abstract Input Domain~\inputdomain.} 
    The \emph{abstract input domain} is an abstract word domain 
    that maps between abstract words~\inputwords and concrete words 
    over the concrete input alphabet~\coninalphabet.
    It provides a denotation function $\sem{\cdot}_\inputcomp: 
    \inputwords \rightarrow 2^{\coninalphabet^*}$
    to map between an abstract word and a set of concrete (and finite) words.
    We assume the lattice of abstract words to be distributive and 
    complemented, that is, to be dual~\cite{Stone1936} to a \emph{Boolean algebra}.
    An abstract domain with lattice-valued regular expression~\cite{MidtgaardEtal2016}
    would be an example of an abstract input domain.
    \label{concept:inputdomain}%

    \label{concept:outalphabet}%
\item \emph{Abstract Output Domain~\outputdomain.} 
    The \emph{abstract output domain} is an abstract word domain that 
    defines the abstract output words~\outputwords and their relationship.
    Its denotation function $\sem{\cdot}_\outputcomp: \outputwords \rightarrow 2^{\conoutalphabet^\infty}$
    maps between an abstract output word and the corresponding 
    set of concrete output words over the concrete output alphabet~\conoutalphabet.
    An instance of an abstract output domain could, for example,
    use antichains~\cite{AbdullaEtal2010} for word inclusion checks.

\item \emph{Initial Transducer State~$\atmtinit \in 2^{\atmtstates \rightarrow \outputwords}$.} 
    The (non-empty) map~$\atmtinit$ characterizes the \emph{initial transducer state}.
    The pairing of control states with outputs is needed, since already
    the transitions that leave the initial state can be $\epsilon$-moves
    that are annotated with an output, 
    and it must be possible to eliminate those moves without
    affecting the semantics of the transducer.

\item \emph{Final Control States~$\atmtfinal \subseteq \atmtstates$.} 
    The set~$\atmtfinal$ defines the final (accepting)~control states. 
    This set can be empty, for example, if the transducer is not intended 
    to operate as a classical acceptor, that is, if the focus is on the 
    intermediate languages.

\item \emph{Transition Relation~$\atmtrelation \subseteq \Delta$.}
    The \emph{transition relation} defines the set of transitions 
    that are possible between the different control states.
    Given a \emph{transducer transition}~$(q, \abstractword, q', \outputword) \in \atmtrelation$,
    with $\Delta = \atmtstates \times \inputwords \times \atmtstates \times \outputwords$,
    both the abstract \emph{transition input word}~\abstractword and the 
    abstract \emph{transition output word}~\outputword can be the 
    abstract epsilon word, which is used to implement the functionality 
    of an $\epsilon$-NFA. 
    The abstract input word~$\abstractword$ must never be the abstract
    bottom word, that is, $\sem{\abstractword}_\inputword \neq \emptyset$.
    Having the empty word as output signals that the matching process 
    must stop for the given abstract input word---nevertheless, 
    there can be another transition from the same state~$q$ that 
    has an intersecting abstract input word
    which can cancel out this effect.
\end{itemize}
\noindent
The set of all transducers is denoted by~$\transducers$, with the 
subset~$\transducers_{\inputdomain \times \outputdomain} \subseteq \transducers$ 
of transducers that transduce from words from an abstract input 
domain~$\inputdomain$ to those from an abstract output 
domain~$\outputdomain$.
\end{definition}

\noindent
%
%

\subsection{State Types}
The set of control states~$\atmtstates$ of an abstract transducer 
implicitly contains two special states that are entered under
certain conditions or are used by algorithms that operate
on abstract transducers:

\begin{definition}[Trap State]    The \emph{trap state} or \emph{inactivity signaling state} is a 
    special control state~$\qinactive$ that can be entered
    to signal that the analysis should continue from that point on,
    but the transducer will no more contribute to the analysis process.
    We assume that this state is implicitly present for each transducer,
    that is, $\qinactive \in \atmtstates$ 
    and~$(\qinactive, \top, \qinactive, \outputword_\epsilon) \in \atmtrelation$,
    with~$\sem{\outputword_\epsilon}_\outputcomp = \{ \epsilon \}$.
\end{definition}

\noindent
The trap state is entered if no more transitions to move
are left, but the analysis should still continue from that point on.
This state is important for configurations of analyses that track 
automata or transducers with a non-stuttering semantics, that is, 
that do not stay in the same state if no transition matches.
We define another, similar, control state:

\begin{definition}[Bottom Control State]    A \emph{bottom control state} or \emph{unreachable control state}
    is a special control state~$\qbot \in \atmtstates$
    that has no leaving transitions and is not an accepting state, 
    that is, $(\qbot, \cdot, \cdot, \cdot) \not \in \atmtrelation$
    and $\qbot \not \in \atmtfinal$,
    while we assume this state to be present for all transducers
    implicitly in their set of control states~$\atmtstates$.
\end{definition}

\noindent
The core of an abstract transducer is its transition relation,
which defines the possible transitions between control states 
and the output to produce on these transitions.
The result of a state transition is a new transducer state:
\begin{definition}[Transducer State]    A \emph{transducer state}~$\transducerstate \in \transducerstates$,
    with~$\transducerstates = 2^{\atmtstates \rightarrow \outputwords}$,
    is map~$\transducerstate: \atmtstates \rightarrow \outputwords$ 
    from control states to abstract output words.
    Typically, a transducer state is the result of running the abstract
    transducer for a given input, starting in the initial transducer 
    state~$\atmtinit \in \transducerstates$.
\end{definition}

\subsection{Mealy and Moore}
We formalize abstract transducers as Mealy-style~\cite{Mealy} 
finite-state machines. 
Nevertheless, also a Moore-style~\cite{Moore} 
representation is possible:
\nopagebreak[4]
\newcommand{\outputlabeling}{\ensuremath{\lambda}\xspace}%
\begin{definition}[Moore-style Abstract Transducer]A \emph{Moore-style abstract transducer} is an abstract transducer
that emits its outputs not on transitions between control states
but \emph{active control states}. That is,
it is defined by the tuple~$\transducer^\text{Moore} =  
    (\atmtstates, \inputdomain, \outputdomain, 
        \atmtstates_0, \atmtfinal, \atmtrelation, \outrelation)$.
This form of abstract transducer has a 
control transition relation~$\atmtrelation \subseteq \atmtstates \times \abstinputwords \times \atmtstates$
and uses a state-output labeling function~$\outputlabeling: \atmtstates \rightarrow \outputwords$
to map abstract output words to control states.
Furthermore, this style of abstract transducer has 
a set~$\atmtstates_0$ of initial control states.
\end{definition}

\noindent
A Moore-style abstract transducer allows to represent
an abstract reachability graph easily.
For this work, we prefer the Mealy-style formalization
of abstract transducers because they require fewer states and 
are fit well for sharing syntactic task artifacts~(program 
fragments for weaving).

After we have defined the components of an abstract transducer,
we continue in following subsections with the description 
of their semantics.

\subsection{Lookaheads and Graph Matching}
Annotating a transition of an abstract transducer with an 
abstract input word that maps to at least 
one concrete word that is longer than one letter, specifies a lookahead.
The possibility of conducting lookaheads is essential if
a transition should produce a particular output only if the 
remaining word to process has a specific word as its prefix.
Consider the following example:
\begin{example}
Assume that the transducer is in control state~$q \in \atmtstates$.
Given a concrete input word~$\word = \langle \sigma_1, \ldots, 
    \sigma_n \rangle \in \inputalphabet^*$, a transducer transition 
$(q, \abstractword, q', \outputword) \in \atmtrelation$, with 
$\sem{\abstractword}_\inputcomp = 
\{ \langle x, `e`, `d` \rangle \; | \; x \in \inputalphabet \}$ 
will only match if~$\sigma_2 = `e` \land \sigma_3 = `d`$
and will then produce the output~$\outputword$.
\end{example}

\noindent
We characterize the lookahead of a transducer transition
by a number:
\nopagebreak[4]
\label{concept:lookahead}%
\begin{definition}[Transition Lookahead]The \emph{lookahead}~$\lookahead(\tau) \in \Natsz$ of a 
transition~$\tau = (q, \abstractword, \cdot, \cdot) \in \atmtrelation$
is $\lookahead(\tau) = 0$ if the input language is either the 
abstract epsilon word or the abstract bottom word, 
otherwise it is defined as~$\lookahead(\tau) = \text{max} \; 
\{ |\word| \; | \; \word\in \sem{\abstractword}_\inputcomp \} - 1$.
\end{definition}
\noindent
The lookahead of an abstract transducer is defined by the maximal 
lookahead that is conducted on one of its transitions, that is:
\nopagebreak[4]
\begin{definition}[Transducer Lookahead]The \emph{lookahead of an abstract transducer}~$\lookahead(\transducer) \in \Natsz$
is the maximal lookahead of any of its transition.
That is, $\lookahead(\transducer) = \text{max} \; 
\{ \lookahead(\tau) \; | \; \tau \in \atmtrelation\}$,
where~$\atmtrelation$ is the transition relation of transducer~\transducer.
\end{definition}

\begin{wrapfigure}[12]{r}{0.3\linewidth}%
\includegraphics[width=0.3\textwidth]{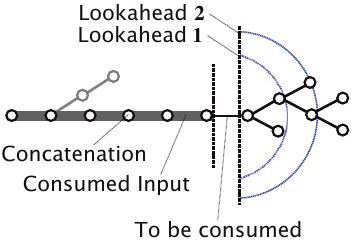}%
\caption{Matching}%
\label{fig:closure:elim}%
\end{wrapfigure}%
One can execute an abstract transducer on a rooted and directed 
graph instead of a particular input word---one word corresponds 
to a list or a sequence of letters. 
Each edge of the graph that we match is labeled with a letter. 
Words are formed by concatenating all letters on the graph edges 
that get traversed during the matching process, 
starting from the root node of the graph.
Figure~\ref{fig:closure:elim} provides an intuition of
the matching process.
\noindent
In this work, we restrict the graph matching process to 
disjunctive tree matching, defined by:
\begin{definition}[Disjunctive Tree Matching]A tree matching procedure is called to be \emph{disjunctive} 
if not several input branches that follow from a particular point 
on have to satisfy specific criteria. 
That is, if only \emph{one} of the input words that follow (on that the 
lookahead is conducted), must satisfy a given criterion.
\end{definition}
\noindent
To allow for matching based on the full expressiveness of regular
tree expressions (several of the input words might have to satisfy
a specific criterion), the abstract transducer's abstract input domain 
has to be lifted from an abstract word domain to an abstract
\emph{language} domain---see Sect.~\ref{sec:chars:words:langs}. 
We keep this extension of abstract transducers for future work.

\subsection{Epsilon Closure}
\label{concept:at:epsilonclosure}%

An established practice~\cite{HopcroftEtal2003,Sipser} in automata 
theory and its application is to use automata with transitions 
that are annotated with an empty-word symbol~$\epsilon$.
This was, first and foremost, introduced as a convenience feature 
to describe automata and its transition relation 
in a more concise fashion.
Abstract transducers allow to annotate transitions with
the abstract epsilon word~$\abstractepsilon$
to provide similar semantics and convenience:
\nopagebreak[4]
\begin{definition}[$\epsilon$-Move]An $\epsilon$-move (or \emph{$\epsilon$-transition}) is an 
automaton transition (or transducer transition) 
$(q, \abstractword, q', \outputword) \in \atmtrelation$
that is annotated with the abstract epsilon word~$\abstractepsilon$ 
as its input, that is, 
$\sem{\abstractword}_\inputcomp = \sem{\abstractepsilon}_\inputcomp = \{\epsilon\}$.
\end{definition}

\noindent
Some algorithms might not be able to deal with transducers that
have $\epsilon$-moves---or they might be more sophisticated in 
their presence---but only with those transducers from that 
all $\epsilon$-moves were eliminated. 
We define abstract transducers without $\epsilon$-moves as:
\nopagebreak[4]
\begin{definition}[Input-$\epsilon$-Free]An abstract transducer is said to be~\emph{input-$\epsilon$-free} 
if it does not have any transition based on an $\epsilon$-move,
that is, $(\cdot, \abstractword_\epsilon, \cdot, \cdot) \not \in \atmtrelation$,
with~$\sem{\abstractword_\epsilon}_\inputcomp = \{ \epsilon \}$.
\end{definition}

\noindent
The presence of $\epsilon$-moves can lead to loops
thereof, which is vital for expressing complex outputs, 
for example, to describe the control-flow of Turing-complete 
programs---assuming that each move emits a program operation 
to conduct as output.

\begin{definition}[$\epsilon$-Loop]An $\epsilon$-loop is any sequence of $\epsilon$-moves that 
starts in a control state~$q_k$ and could include this control
state~$q_k$ infinitely often in such a sequence.
More formally, an $\epsilon$-loop is a sequence~%
$\bar{\tau} = \langle \tau_1, \ldots, \tau_n \rangle \in \Delta^\infty$ 
of $\epsilon$-moves that is well-founded in the
transition relation~\atmtrelation and there exists a transducer 
transition~$\tau_i = (q, \cdot, \cdot, \cdot) \in \bar{\tau}$ for 
which the source state~$q$ is precisely the destination state~$q'$ 
of a transducer transition~$\tau_j = (\cdot, \cdot, q', \cdot) \in \bar{\tau}$, 
with~$i \leq j$.
\end{definition}

\noindent
From the definition of $\epsilon$-moves follows
the definition of the $\epsilon$-closure~\cite{Sipser}. 
Intuitively speaking the $\epsilon$-closure of a control state~$q$ 
is the set of control states that become instantly and 
simultaneously (parallel) active if state~$q$ becomes active.

\begin{definition}[Epsilon Closure]The~\emph{epsilon closure}~$\epsclosure: \atmtstates \rightarrow 2^\atmtstates$
of a state~$q \in \atmtstates$ is the set~$\epsclosure(q) \subseteq \atmtstates$ 
of states that can get reached transitively from state~$q$ by 
only following $\epsilon$-moves~\cite{Sipser}.
The bottom state~$\qbot$ is added if the epsilon closure 
includes an $\epsilon$-loop from which no control state is 
reachable with no $\epsilon$-move leaving.
\end{definition}

\noindent
The transition relation of an abstract transducer can contain
sequences~$\{ (q_1, \epsilonword, q_2, \allowbreak \outputword_1),
\allowbreak (q_2, \epsilonword, q_3, \outputword_2) \}
\subseteq \atmtrelation$ 
of $\epsilon$-moves but not each control state that is reached 
within such a sequence might have non-$\epsilon$-moves leaving 
in the transition relation.
We therefore introduce the notion of closure termination states:
\nopagebreak[4]
\begin{definition}[Closure Termination States]\label{def:cts}
The \emph{closure termination states}~$\closuretermstates: \atmtstates \rightarrow 2^\atmtstates$
of a given state~$q$ are both the states (1)~in the epsilon 
closure~$\epsclosure(q)$ from which no $\epsilon$-move leaves
and (2)~states within the closure that are accepting, that is, 
$\closuretermstates(q) = \{ q' \; | \; q' \in \epsclosure(q) 
\land (q', \epsilonword, \cdot, \cdot) \not \in \atmtrelation \} 
\cup (\epsclosure(q) \cap \atmtfinal)$.
\end{definition}

\noindent
Each transducer transition between the control states from 
an $\epsilon$-closure can be mapped to a set of closure termination states:

\newcommand{\termstates}{\ensuremath{\Delta_\Omega}\xspace}
\begin{definition}[Termination State Mapping]The \emph{termination state mapping} is a 
map~$\termstates: \Delta \rightarrow 2^\atmtstates$
that maps a given transducer transition to the set of closure
termination states that are reachable. Given a control state~$q \in \atmtstates$,
the result is the empty set~$\emptyset$ if no $\epsilon$-move
leaves state~$q$; it is the bottom state~$\qbot$ if there is 
not any other termination state.
\end{definition}

\noindent
Since also each transition within an epsilon closure can produce
an output, we introduce the notion of \emph{concrete language 
on termination}. 
This notion reflects with which output words the different
closure termination states can be reached:

\newcommand{\epspathset}{\ensuremath{\Omega}\xspace}%
\newcommand{\closurelang}{\ensuremath{\Omega}\xspace}%
\begin{definition}[Concrete Language on Termination]The \emph{concrete language on termination} 
$\closurelang: \atmtstates \times \atmtstates \rightarrow 2^\outwords$ 
for a given pair~$(q, q_\Omega)$ describes the concrete output 
language (a set of concrete words) that can be produced starting 
in control state~$q$ and that terminates with a closure termination 
state~$q_\Omega \in \closuretermstates(q)$.
More formally, let $\hat{\tau} = \{ \bar{\tau}_1, \ldots \} \subseteq \Delta^\infty$ 
be the set of all well-founded sequences of transducer transitions 
between control state~$q$ and the termination state~$q_\Omega$,
with~$\bar{\tau_i} = \langle \tau_1, \ldots \rangle$
and~$\tau_i = (q, \abstractword_i, q', \outputword_i) \in \atmtrelation$.
The concrete output language~$\sem{\bar{\tau_i}}$ of a 
sequence~$\bar{\tau_i}$ is the concatenation 
$\sem{\outputword_1}_\outputcomp \concat \ldots$
of the concretizations of all abstract output words~$\outputword_i$ 
that are emitted along it.
That is, the concrete output language $\closurelang(q,q_\Omega)$
is the union~$\bigcup_{\bar{\tau_i} \in \hat{\tau}} \; \sem{\bar{\tau_i}}$.
\end{definition}

\begin{definition}[Concrete Closure Language]The \emph{concrete closure language}~$\closurelang(q) \subseteq \outwords$ 
of a given control state~$q$ and its $\epsilon$-closure is the 
set of concrete output words that is produced while making transitions
along the $\epsilon$-moves between states in the closure.
More precisely, it is the join of concrete languages on termination,
that is, $\closurelang(q) = \bigcup \; \{ \word \in \closurelang(q, q_\Omega)  \; | \; q_\Omega \in \closuretermstates(q) \}$.
\end{definition}

\noindent
In our applications of abstract transducers, we use the (anonymous) 
states and transitions in the epsilon closure as a tool for 
expressing relational outputs.
Please note, that also $\epsilon$-moves that lead to a \emph{dead-end} 
are relevant and \emph{must not be eliminated}---what is 
done for some applications~\cite{DAntoniVeanes2017}---%
because the output might be relevant for the analysis task,
and the soundness of the produced result, for which the 
transducer is executed.
\nopagebreak[4]
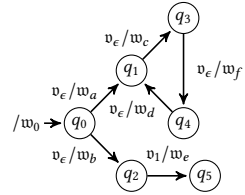
\begin{wrapfigure}[9]{r}{0.3\linewidth}%
\begin{tikzpicture}[node distance=8mm, scale=0.7, transform shape]
\node[] (es) {};
\node[cfastate,right of=es] (e0) {$q_0$};
\node[cfastate,above right=of e0] (e1) {$q_1$};
\node[cfastate,below right=of e0] (e2) {$q_2$};
\node[cfastate,above right=of e1] (e3) {$q_3$};
\node[cfastate,below right=of e1] (e4) {$q_4$};
\node[cfastate,right=of e2] (e5) {$q_5$};
\path[trans] (es) edge node [pos=.5,label=left:{$/\outputword_0$}] {} (e0);
\path[trans] (e0) edge node [pos=.7,label=left:{$\abstractepsilon/\outputword_a$}] {} (e1);
\path[trans] (e1) edge node [pos=.7,label=left:{$\abstractepsilon/\outputword_c$}] {} (e3);
\path[trans] (e0) edge node [pos=.7,label=left:{$\abstractepsilon/\outputword_b$}] {} (e2);
\path[trans] (e4) edge node [pos=.1,label=left:{$\abstractepsilon/\outputword_d$}] {} (e1);
\path[trans] (e3) edge node [pos=.5,label=right:{$\abstractepsilon/\outputword_f$}] {} (e4);
\path[trans] (e2) edge node [pos=.5,label=above:{$\abstinputword_1/\outputword_e$}] {} (e5);
\end{tikzpicture}
\caption{With $\epsilon$-loop}
\label{fig:closure}
\end{wrapfigure}
\begin{example}
Figure~$\ref{fig:closure}$ illustrates an example transducer:
The $\epsilon$-closure of control state~$q_0$ is the 
set~$\epsclosure(q_0) = \{ q_0, q_1, q_2, q_3, q_4, \qbot \}$,
for state~$q_2$, the closure~$\epsclosure(q_2) = \{ q_2 \}$ 
does not contain additional states.
State~$q_0$ has the set of closure termination 
states~$\closuretermstates(q_0) = \{ q_2, \qbot \}$, and 
state~$q_1$ has $\closuretermstates(q_1) = \{ \qbot \}$, 
that is, no other termination state is reachable.
The transitions between states~$\{ q_1, q_3, q_4\}$ form an $\epsilon$-loop.
\end{example}

\noindent
Given a control state~$q \in \atmtstates$, the semantics of 
$\epsilon$-moves implies that with reaching state~$q$, actually all 
states in~$Q_t = \closuretermstates(q)$ are reached immediately.
That is, also all output on the transitions from $q$ to a state 
in $Q_t$ is produced immediately, resulting in---possibly 
exponentially many and infinitely long---words~$\subseteq \outwords$ 
over the output alphabet~$\outalphabet$.

\subsection{Output Closure}
\label{sec:abstepsclosure}%
\label{concept:abstepsclosure}%

Previous section describes the epsilon closure of abstract
transducers; in contrast to established transducer concepts,
we also address $\epsilon$-moves that are annotated with non-empty
outputs, and use them as tool to express complex output languages,
with possibly exponentially many and infinitely long words,
in a convenient fashion.
%
%
When executing or reducing (minimizing) abstract 
transducers, means for collecting, aggregating, and possibly 
abstracting the output on these transitions are needed.
Given a control state~$q \in \atmtstates$, the goal of this 
summarization process is to provide an abstract output 
word~$\outputword_\Omega \in \outputwords$ 
for each of its closure termination states~$q_\Omega \in \closuretermstates(q)$
that overapproximates the concrete closure language, that is,
$\closurelang(q, q_\Omega) \subseteq \sem{\outputword_\Omega}_\outputcomp$---%
which \emph{can} lead to a loss of information.
The computation of this closure is done in a corresponding operator:
\nopagebreak[4]
\begin{definition}[Abstract Output Closure]The \emph{abstract output closure} of a given control 
state~$q \in \atmtstates$ is a finite overapproximation of the 
concrete closure language of each of its closure termination states;
it is a map of closure termination states of~$q$ to 
abstract output words, which summarizes the corresponding
closure output languages: 
$\outclosure: (\atmtstates \times \outputwords) \rightarrow 2^{\atmtstates \rightarrow \outputwords}$.
A call~$\outclosure(q, \outputword_0)$, with an initial abstract
output word~$\outputword_0$, returns a map~$\{ (q_t, \outputword_t) \; | 
\; q_t \in \closuretermstates(q) 
\land \sem{\outputword_t}_\outputcomp \subseteq \sem{\outputword_0}_\outputcomp \concat \closurelang(q,q_t) \}$.
\end{definition}

\noindent
We extend the abstract output closure operator~$\outclosure$ to sets:

\begin{definition}[Abstract Output Closure]The abstract output closure of a given set of transducer states
$\widehat{\outclosure}: 2^{\atmtstates \times \outputwords} 
\rightarrow 2^{\atmtstates \rightarrow \outputwords}$ is
defined as~$\widehat{\outclosure}(S) = \imagejoin \bigcup \;
\{ (q_\Omega, \outputword_\Omega) \; | \allowbreak 
\; \allowbreak (q, \outputword_0) \in S \land \allowbreak
(q_\Omega, \outputword_\Omega) \in \outclosure(q, \outputword_0) \} $.
\end{definition}

\noindent
Actual implementations of an abstract output closure operator
can be provided, for example, based on abstract interpretation, 
or based on techniques from automata theory.
Even transducers can be used~\cite{PredaEtal2016} to compute
abstractions of languages, in our case, the concrete output 
languages that are produced in the $\epsilon$-closure.
We give two examples of implementations:

\newcommand{\closuretrans}{\ensuremath{\textsf{closuretrans}}\xspace}

\subsubsection{Joining Closure}
The first abstract output closure operator~$\outclosure_\join$ 
joins all abstract output words that can be found on transitions in 
the epsilon closure from control state~$q$ that are mapped to 
the same closure termination state.
Let us assume that there is an operation~$\closuretrans: \atmtstates \times \atmtstates \rightarrow 2^\Delta$
that, given a pair of control states~$q, q_\Omega \in \atmtstates$,
returns all transitions from the transition relation~$\atmtrelation$ 
that are in the epsilon closure~$\epsclosure(q)$ and are mapped 
to a closure termination state~$q_\Omega$.
Then, we can define the closure operator as follows: 
$\outclosure_\join(q, \outputword_0) = \allowbreak 
\{ (q_\Omega, \outputword_0 \join_\outputcomp \bigsqcup_\outputcomp \{ \outputword \; | \allowbreak
\; (\cdot, \cdot, \cdot, \outputword) \in \closuretrans(q, q_\Omega)\}) \; | \allowbreak 
\; q_\Omega \in \closuretermstates(q) \} $.
This operator produces an overapproximation of the concrete output language.
The resulting abstraction does neither preserve information on
the flow nor is path information kept.


\subsubsection{Regular Closure}
Another example of an output closure operator is~$\outclosure_\infty$.
Here, we assume that the abstract output words can be 
described based on an abstract domain of $\infty$-regular 
languages~\cite{LodingTollkoetter2016}, with a corresponding 
lattice thereof.
Rules for transforming automata into regular expressions can be 
applied~\cite{LodingTollkoetter2016}: The result for the 
transducer in Fig.~\ref{fig:closure} is
${\outclosure_{\infty}}(q_0, \outputword_0) = 
\{ (q_2, \outputword_0 \concat \outputword_b), 
   (\qbot, \outputword_0 \concat \outputword_a \concat (\outputword_c \concat \outputword_f \concat \outputword_d)^\omega) \}$.
This type of output closure is lossless. Nevertheless, not 
all applications require this level of detail.

\subsection{Runs}

We now define runs of abstract transducers and illustrate
how they are conducted for given inputs.
All runs of an abstract transducer start from the initial 
transducer state:
\nopagebreak[4]
\begin{definition}[Abstract Transducer Run]A \emph{run} of an abstract transducer on a concrete input 
word~$\word = \langle \sigma_1, \allowbreak \ldots, \allowbreak \sigma_n \rangle \in \coninalphabet^*$ 
and a lookahead~$\inputwordset \subseteq \words$
is a sequence of transducer state 
transitions~$\atmtinit \trans{\abstinputword_1/\outputword_1} 
\ldots \trans{\abstinputword_n/\outputword_n} \transducerstate_n$,
also denoted by~$\langle \atmtinit, \ldots, \transducerstate_n \rangle$
in case the actual transducer transitions are irrelevant for the
discussion.
A run always starts in the initial transducer 
state~$\atmtinit \in \transducerstates$,
is well-founded in the transition relation~$\atmtrelation$,
and all transitions along the input match, that is, the 
quotienting $\quotientof{(\abst{ \{\langle \sigma_i, \ldots, \sigma_n \rangle \} \concat \inputwordset}_\inputcomp)}{\abstinputword_i} \neq \emptylang$ 
does not result in the abstract bottom word.
\end{definition}

\noindent
Before we continue to define feasible and accepting runs of 
an abstract transducer, we define the abstract output of a run:
\nopagebreak[4]
\begin{definition}[Abstract Run Output]The \emph{abstract output} of a run~$\transducerstate_0 \trans{\abstinputword_1/\outputword_1} \ldots \trans{\abstinputword_n/\outputword_n} \transducerstate_n$ 
is the concatenation of the subsequent abstract output words
$\outputword_\concat = \outputword_0 \concat \outputword_1 \concat \ldots \concat \outputword_n$.
The abstract output word~$\outputword_0$ is one abstract output 
word from the initial transducer state, that is, there exists a 
pair~$(\cdot, \outputword_0) \in \atmtinit$.
\end{definition}

\noindent
The output of a abstract transducer run is essential for
the definition of feasible transducer runs:

\newcommand{\wordfeasible}{\ensuremath{\textsf{feasible}}\xspace}%
\begin{definition}[Feasible Run]A run is called \emph{feasible} if and only if its abstract 
output~$\outputword_\concat$
is not the bottom element~$\bot$, that is, if and only if
$\sem{\outputword_\concat}_\outputcomp \neq \emptyset$.
The set of all concrete inputs (with lookaheads) that result
in a feasible run on an abstract transducer~$\transducer$ 
defines the function~$\wordfeasible_\transducer: \inputalphabet^* \times 2^{\inputalphabet^*}\rightarrow \Bools$.
\end{definition}

\noindent
Abstract transducers can also operate as acceptors and define a 
set of inputs to be accepted. 
We first define the notion of an accepting run and define
the accepted input language later:

\begin{definition}[Accepting Run]A run~$\langle \transducerstate_0, \ldots, \transducerstate_n \rangle$ 
is called to be \emph{accepting} if it is feasible and 
its last transducer state contains an accepting (final) control state,
that is, if and only if $(q_n, \outputword_n) \in \transducerstate_n$, with $q_n \in \atmtfinal$
and $\outputword_n \neq \bot$.
\end{definition}

\noindent
In general, an abstract transducer is a nondeterministic
automaton, nevertheless it can be deterministic if it satisfies
following criterion:
\nopagebreak[4]
\begin{definition}[Deterministic Abstract Transducer]We call an abstract transducer~\emph{deterministic} if and only
if it does not allow a run~$\bar{\iota} = \langle \transducerstate_0, \ldots \transducerstate_n \rangle$ 
with a transducer state~$\transducerstate_i \in \bar{\iota}$ 
that consists of more than one element, that is, 
$\forall \transducerstate_i \in \bar{\iota}: |\transducerstate_i| \leq 1$.
\end{definition}

We now continue with an operational perspective on the runs
of an abstract transducer.
Given a \emph{concrete input word}~$\inputword \in \coninalphabet^*$
based on the concrete input alphabet~\coninalphabet and a 
set~$\inputwordset \subseteq \coninalphabet^*$ of words that can 
follow to this word~(used for the lookahead), which output does 
the transducer produce and does processing the word terminate in 
an accepting control state?
Since a concrete input word can be represented as an abstract
word, and we consider this the more general case, we describe
runs based on abstract input words:
A given concrete input word~$\concinputword \in \inputalphabet^*$
can be transformed to an abstract input word by applying the 
abstraction operator such that we end up in an abstract 
word~$\abstinputword = \abst{\{ \concinputword \}}_\inputcomp$, 
with $\sem{\abstinputword}_\inputcomp = \{ \concinputword \}$.

\newcommand{\runconfigs}{\ensuremath{R}\xspace}%
\newcommand{\runconfig}{\ensuremath{r}\xspace}%
\newcommand{\runfor}{\ensuremath{\mathsf{run}_0}\xspace}%
\newcommand{\runfrom}{\ensuremath{\mathsf{run}}\xspace}%
\newcommand{\runfromset}{\ensuremath{\hat{\mathsf{run}}}\xspace}%
\begin{definition}[Run]The function~$\runfrom_\transducer: \atmtstates \times \outputwords 
\times \abstinputwords \times \abstinputwords 
\rightarrow 2^{\atmtstates \rightarrow \outputwords}$ 
conducts a run starting from a control state~$q \in \atmtstates$,
an initial abstract output word~$\outputword \in \outputwords$, 
an abstract input word~$\abstinputword \in \abstinputwords$, 
with~$\abstinputword \neq \bot$,
and an abstract word~$\abstinputword_\lookahead \in \abstinputwords$
that describes the lookahead that must be satisfied:
$$
\runfrom_\transducer(q, \outputword, \abstinputword, \abstinputword_\lookahead) = \begin{dcases*}
    \{ \; (q, \outputword) \; \}
     & \text{if $\abstinputword = \abstinputword_\epsilon$} \\
     \imagejoin \; \bigcup \; \{ \; \runfrom_\transducer({q'', \outputword \concat \outputword'', \tailof{\abstinputword}, \abstinputword_\lookahead}) 
     \;|\; \\ \quad (q, \abstinputword_\tau, q', \outputword') \in \atmtrelation 
        \\ \quad \land \; (q'', \outputword'') \in \outclosure(q', \outputword') 
        \\ \quad \land \; \quotientof{(\abstinputword \concat \abstinputword_\lookahead)}{\abstinputword_\tau} \neq \bot 
        \\ \quad \land \; \quotientof{(\headof{\abstinputword})}{\headof{\abstinputword_\tau}} \neq \bot \; \} & \text{otherwise} \\
\end{dcases*}
$$

\noindent
The function~\runfrom terminates its recursion if the abstract 
input word is the bottom element. 
The recursive call to~\runfrom is done for the tail of the abstract
input word---which ensures termination---in case a transition that 
leaves the given control state~$q$ matched the input.
\end{definition}
\noindent
We extend this function to 
$\runfromset_\transducer: 2^{\atmtstates \rightarrow \outputwords} \times \abstinputwords \times \abstinputwords \rightarrow 2^{\atmtstates \rightarrow \outputwords}$,
which starts from a transducer state, and we define it as follows:
$$
\runfromset_\transducer(\transducerstate, \abstinputword, \abstinputword_\lookahead) = \imagejoin \bigcup_{
    \quad (q, \outputword) \in \transducerstate} 
    \runfrom_\transducer(q, \outputword, \abstinputword, \abstinputword_\lookahead)
$$

\noindent
The transducer state to start from is omitted if it is the
abstract transducer`s initial transducer state~$\atmtinit$,
that is, $\runfromset_\transducer(\abstinputword, \abstinputword_\lookahead) 
= \runfromset_\transducer(\atmtinit, \abstinputword, \abstinputword_\lookahead)$.
Given a concrete input word~$\bar{\sigma} \in \inputalphabet^*$ 
and a corresponding set of concrete words~$\hat{\sigma} \subseteq \inputalphabet^*$
for the lookahead,
we write $\runfromset_\transducer(\bar{\sigma}, \hat{\sigma})$ 
as an abbreviation for $\runfromset_\transducer(\atmtinit, \bar{\sigma}, \hat{\sigma})$, 
which is as an abbreviation 
for~$\runfromset_\transducer(\atmtinit, \abst{\{ \bar{\sigma} \}}_\inputcomp, 
\abst{\hat{\sigma}}_\inputcomp)$.

\subsection{Languages and Transductions}
Contrary to other types of finite state transducers~\cite{DAntoniVeanes2015}
our abstract transducers distinguish between two type of input languages: 
the intermediate input language and the accepted input language.
\nopagebreak[4]
\begin{definition}[Intermediate Input Language]The \emph{intermediate input language}~$\langin(\transducer) \subseteq 
    \inputalphabet^* \times 2^{\inputalphabet^*}$ of an abstract 
transducer~\transducer is the set of concrete input words for 
that the transducer can conduct feasible runs starting from the 
initial transducer state~$\atmtinit$:
$$\langin(\transducer) = 
\{ \; (\inputword, \inputwordset) \; | \;  
\wordfeasible_\transducer(\inputword, \inputwordset) 
\land \inputword \in \inputalphabet^* \land \inputwordset \subseteq \inputalphabet^*
\; \}.
$$
It follows that each prefix~$\word_p \prefixof \word$ of each 
word~$\word \in \langin(\transducer)$ is also element of the 
intermediate input language, that is, $\word_p \in \langin(\transducer)$.
\end{definition}

\noindent
The accepted input language reflects the established notion
of input language, which is based on the set of words that can
reach a final control state:

\begin{definition}[Accepted Input Language]The \emph{accepted input language}~$\langacc \subseteq \langin$
is the subset of the intermediate input language for which an 
accepting control state~$q \in \atmtfinal$ is reached:
$$\langacc(\transducer) = 
\{ \; (\bar{\sigma}, \hat{\sigma}) \in \langin(\transducer) \; | \;  
(q, \cdot) \in \runfromset_\transducer(\bar{\sigma}, \hat{\sigma}) 
\land q \in \atmtfinal 
\; \}.
$$
\end{definition}

\noindent
Beside the accepted input language, another characteristic of
an abstract transducer is its set of transductions and its set 
of accepting transductions:

\begin{definition}[Transductions]\label{concept:transductions}%
The set of \emph{transductions}~$\transductions(\transducer) \subseteq 
\inputalphabet^* \times 2^{\inputalphabet^*} \times 2^{\outputalphabet^\infty}$
of an abstract transducer~\transducer characterizes both its 
concrete input language and the outputs that are produced for them.
One element~$(\inputword, \bar{\inputalphabet}_\lookahead, 
\bar{\outputalphabet}) \in \transductions(\transducer)$ from this
set is a tuple that consists of a word prefix~$\inputword$ that 
is consumed by a run of the transducer, a set of 
concrete words~$\bar{\outputalphabet} \subseteq \inputalphabet^*$ 
to conduct the lookahead on and that remains to be consumed by
the next transitions of the transducer,
and the set of concrete output words~$\bar{\outalphabet} \subseteq \outputalphabet^\infty$
that are emitted with the consumption of word~$\inputword$---see the
definition of~$\runfromset_\transducer$ for more details:
\begin{align*}
\transductions(\transducer) = \bigcup \; \{ \; 
    & (\bar{\sigma}, \hat{\sigma}, \sem{\outputword}_\outputcomp) \; 
     | \;  (\bar{\sigma}, \hat{\sigma}) \in \langin(\transducer)  \; \\
    & \land (q, \outputword) \in \runfromset_\transducer(\bar{\sigma}, \hat{\sigma}) \; \}.
\end{align*}
\end{definition}

\newcommand{\acctransductions}{\ensuremath{\transductions_\text{acc}}\xspace}%
\begin{definition}[Accepting Transductions]The set of \emph{accepting transductions}~$\acctransductions(\transducer) \subseteq \transductions(\transducer)$ 
is the subset of the transductions of a given abstract transducer~\transducer
that are produced by accepting runs:
\begin{align*}
\transductions(\transducer) = \bigcup \; \{ \; 
    & (\bar{\sigma}, \hat{\sigma}, \sem{\outputword}_\outputcomp) \; 
     | \; (\bar{\sigma}, \hat{\sigma}) \in \langin(\transducer) \\
    & \land (q, \outputword) \in \runfromset_\transducer(\bar{\sigma}, \hat{\sigma}) \\
    & \land q \in \atmtfinal \; \}.
\end{align*}
The number of accepting transductions is greater
or equal than the number of accepted input words, that is,
$|\langacc(\transducer)| \leq |\acctransductions(\transducer)|$,
because there can be independent concrete output languages for one 
concrete input~$(\bar{\sigma}, \hat{\sigma}) \in \Sigma^* \times 2^{\Sigma^*}$.
\end{definition}

\noindent
In combination, the set of transductions and the set of accepted 
transductions determine if two abstract transducers are equivalent 
to each other:

\begin{definition}[Equivalence] 
Two abstract transducers~$\transducer_1, \transducer_2 \in \transducers$
are called \emph{equivalent} $\transducer_1 \equiv \transducer_2$
to each other
if and only if both have the same set of transductions and the 
same set of accepting transductions, that is, if and only 
if~$\transductions(\transducer_1) = \transductions(\transducer_2)$
and~$\acctransductions(\transducer_1) = \acctransductions(\transducer_2)$.
\end{definition}


\noindent
Based on the notion of equality, we can define different operations, 
for example, reduction or $\epsilon$-elimination.
We start by defining a more fundamental one: The union of two abstract transducers.
The union is constructed similar to the union of $\epsilon$-NFAs, 
with the exception that no $\epsilon$-moves are added; we take 
advantage of the fact that the initial transducer state is a set:

\begin{definition}[Union]Given two abstract transducers~$\transducer_1, \transducer_2 \in 
\transducers_{\inputdomain \times \outputdomain}$
that both have the same abstract input domain~\inputdomain
and the same abstract output domain~\outputdomain,
such that~$\transducer_1 =  (\atmtstates_1, \inputdomain, \outputdomain, 
                \atmtinit_1, \atmtfinal_1, \atmtrelation_1)$
and~$\transducer_2 =  (\atmtstates_2, \inputdomain, \outputdomain, 
                \atmtinit_2, \atmtfinal_2, \atmtrelation_2)$.
The \emph{union}~$\cup: \transducers \times \transducers \rightarrow \transducers$
of two abstract transducers results in a new abstract 
transducer~$\transducer_\cup = \transducer_1 \cup \transducer_2$ 
that maintains exactly both the union of the set of transductions and
the set of accepting transductions, that is,
$\transductions(\transducer_\cup) = \transductions(\transducer_1) 
\cup \transductions(\transducer_2)$ and
$\acctransductions(\transducer_\cup) = \acctransductions(\transducer_1) 
\cup \acctransductions(\transducer_2)$.
We define the union as
    $\transducer_\cup =  \cup(\transducer_1, \transducer_2) = (\atmtstates_1 \cup \atmtstates_2, \inputdomain, \outputdomain, 
           \atmtinit_1 \cup \atmtinit_2, \atmtfinal_1 \cup \atmtfinal_2, \atmtrelation_1 \cup \atmtrelation_2)$.
\end{definition}

\subsection{Elimination of $\epsilon$-Moves}
Since $\epsilon$-moves are considered to be a convenience feature, 
eliminating them without losing any output must be possible---that is, 
without altering the semantics of the transducer.
The $\epsilon$-closure can allow sequences of state
transitions of infinite length, that is, a means to encode this infinite 
information into one~(finite) output symbol is needed.
An algorithm for computing abstract output closures provides such a means.

For the design of an $\epsilon$-elimination algorithm, it is 
important to note that all states in the $\epsilon$-closure 
of a control state become active when it is entered.
This implies that then also the output that is produced along 
with these $\epsilon$-moves must be emitted:
Existing algorithms for $\epsilon$-elimination
are not applicable to abstract transducers.
An algorithm for eliminating $\epsilon$-moves from an abstract
transducer~$\transducer_\epsilon$ must ensure that the resulting
transducer~$\transducer$ is equivalent~$\transducer_\epsilon \equiv \transducer$.
Please note that stuttering transitions must be made explicit and
must be considered to allow a sound elimination of $\epsilon$-moves.

\newcommand{\starttrigger}{\ensuremath{\abstinputword_\text{start}}\xspace}%
\begin{algorithm}[t]
\begin{small}
\caption{$\textsf{elim}{}(\transducer_\epsilon)$}
\label{alg:elim}
\begin{algorithmic}[1]
\vspace{1mm}
\INPUT Abstract transducer~$\transducer_\epsilon = (\atmtstates, \inputdomain, 
    \outputdomain, \atmtinit, \atmtfinal, \atmtrelation) \in \transducers$ \\
\vspace{1mm}
\OUTPUT Abstract transducer~$\transducer \in \transducers$, 
    with~$\transducer_\epsilon \equiv \transducer$
\vspace{1mm}
\newline
\COMMENT{// Sentinel transitions for the initial transducer state, with~$\abstractepsilon \not \equiv \starttrigger$}
\STATE $\atmtrelation_\epsilon = \atmtrelation \cup \{ \; (q_s, \starttrigger, q, \outputword) \; 
    | \; (q, \outputword) \in \atmtinit \; \}$ 
\newline
\COMMENT{// Shortcut $\epsilon$-moves to their termination states}
\STATE $\atmtrelation' = \{ \; (q, \abstractword, q'', \outputword'') \;
    | \; \tau = (q, \abstractword, q', \outputword) \in \atmtrelation_\epsilon \newline
    \quad\quad\quad \land \abstractword \neq \abstractepsilon 
        \land (q'', \outputword'') \in \outclosure(\{ (q', \outputword) \}) \; \}$\label{lst:line:atmtrel}
\newline
\COMMENT{// Reconstruct a new initial transducer state}
\STATE $\atmtinit' = \{ \; (q, \outputword) \; | 
    \; (\cdot, \abstractword, q, \outputword) \in \atmtrelation \land \abstinputword = \starttrigger \; \}$
\newline
\COMMENT{// Reassemble the components to a new abstract transducer}
\STATE \textbf{return} $(\atmtstates, \inputdomain, 
    \outputdomain, \atmtinit', \atmtfinal, \atmtrelation')$
\end{algorithmic}
\vspace{1mm}
\end{small}
\end{algorithm}
Algorithm~\ref{alg:elim} is our approach for eliminating
$\epsilon$-moves from an abstract transducer.
The algorithm constructs a new transition relation,
from which all $\epsilon$-moves are removed by adding shortcuts
to the closure termination states and concatenating
the corresponding closure output language.
\begin{proposition}
Given an abstract transducer~$\transducer_\epsilon$, all its 
$\epsilon$-moves \emph{can} be eliminated without affecting its 
semantics, that is, without affecting either the 
set of transductions or the set of accepting transductions.
The abstract transducer~$\transducer_\epsilon$ can be transformed
into an input-$\epsilon$-free transducer~$\transducer$, 
with $\transducer_\epsilon \equiv \transducer$.
\end{proposition}

\begin{proof}
We prove the proposition by providing an algorithm that conducts 
this transformation while maintaining the set of transductions 
and the set of accepted transductions:
Given an abstract transducer~$\transducer_\epsilon$ that has 
$\epsilon$-moves, Algorithm~\ref{alg:elim}---which we implicitly 
parameterize with the output closure operator~$\outclosure_\infty$---%
produces an abstract transducer~\transducer that is input-$\epsilon$-free
and satisfies~$\transductions(\transducer_\epsilon) = \transductions(\transducer)$
and $\acctransductions(\transducer_\epsilon) = \acctransductions(\transducer)$.
(1)~The transition relation~$\atmtrelation'$, and with it the 
resulting transducer~\transducer, is input-$\epsilon$-free because 
only non-$\epsilon$-moves are added to the transition relation.
(2)~The set of closure termination states, for which~$\outclosure_\infty$
provides a pairing with the corresponding output closure language,
contains all accepting states~(Definition~\ref{def:cts}), that is,
all moves to accepting states are maintained, and with it the 
set of accepting transductions.
(3)~The set of transductions is maintained: The output from 
the epsilon closures, that is, the closure termination languages,
are concatenated to the transitions to the closure termination states. 
\end{proof}


\begin{example}
Given the transducer in Fig.~\ref{fig:closure},
Algorithm~\ref{alg:elim} proceeds as follows:
First, we extend the transition relation with sentinels
and get $\atmtrelation_\epsilon = \{ 
(q_0, \abstractepsilon, q_1, \outputword_a),\allowbreak
(q_1, \abstractepsilon, q_3, \outputword_c),\allowbreak
(q_3, \abstractepsilon, q_4, \outputword_f),\allowbreak
(q_4, \abstractepsilon, q_1, \outputword_d),\allowbreak
(q_0, \abstractepsilon, q_2, \outputword_b),\allowbreak
(q_2, \abstractepsilon, q_5, \outputword_e),\allowbreak
(q_s, \starttrigger, \allowbreak q_0, \abstractepsilon) \}$.
In the next step, $\epsilon$-moves are left out by
adding transitions to the closure termination states and 
concatenating the corresponding closure output languages;
the result is a new transition relation%
~$\atmtrelation' = \{ (q_s, \starttrigger, \qbot, \outputword_\ast),\allowbreak
(q_s, \starttrigger, q_2, \outputword_b),\allowbreak
(q_2, \abstinputword_1, q_5, \outputword_e) \}$,
with~$\outputword_\ast = \outputword_a \concat (\outputword_c \concat \outputword_f \concat \outputword_d)^\omega$.
Then, the initial transducer state is re-constructed from the
relation~$\atmtrelation'$ and we get~$\atmtinit' = \{ 
(\qbot, \outputword_\ast), (q_2, \outputword_b) \}$.
Finally, the transducer is re-assembled and we
get the transducer shown in Fig.~\ref{fig:closure:elim}.
\end{example}

\subsection{Determinization}
A typical operation when dealing with finite state machines is 
the transformation of a nondeterministic automaton into 
a deterministic one. 
This is \emph{not possible} for abstract transducers in general:
The control-flow structure of the state-transitions within the 
$\epsilon$-closure describes different information flows---that 
is, sets of output words that reach different closure termination 
states---as its semantics, which is not the case for classical 
automata and transducers.
For example, a state-space splitting might be intended based on 
the information of the emitted output---different outputs for 
the same input that lead to different control states.
That is, different closure termination states, which can be accepting 
states, can have associated different closure termination languages;
this separation must be maintained---which is also reflected
in our definition of transducer equivalence.

\begin{proposition}
Not every nondeterministic abstract transducer~$\transducer$ can be 
transformed into an equivalent deterministic transducer~$\transducer_d$,
with~$\transducer \equiv \transducer_d$.
\label{prop:det}
\end{proposition}

\begin{proof}
We proof the proposition by counterexample---%
assuming that all abstract transducers can be determinized. 
Given an abstract transducer~\transducer with the set of 
initial transducer
states~$\atmtinit = \{ (q_1, \outputword_1), (q_2, \outputword_2) \}$
and the relation~$\atmtrelation = \{ (q_1, \abstinputword_1, q_3, \outputword_3), 
(q_2, \abstinputword_2, q_4, \outputword_4) \}$,
with~$\langle a \rangle \in \sem{\abstinputword_1}_\inputcomp$ 
and~$\langle b \rangle \in \sem{\abstinputword_2}_\inputcomp$,
it has the set of transductions~$\transductions(\transducer) = \{ 
    (\epsilon, \{ \epsilon \}, \sem{\outputword_1}_\outputcomp),\allowbreak
    (\epsilon, \{ \epsilon \}, \sem{\outputword_2}_\outputcomp),\allowbreak
    (\langle a \rangle, \{ \epsilon \}, \sem{\outputword_1 \concat \outputword_3}_\outputcomp),
    (\langle b \rangle, \{ \epsilon \}, \sem{\outputword_2 \concat \outputword_4}_\outputcomp)
\}$.
A determinized version would have an initial transducer state
with only one element, that is, the initial transducer state 
can be either~$\atmtinit_1 = \{ (q_0, \outputword_1 \join \outputword_2)\}$ 
of a transducer~$\transducer_1$ or~$\atmtinit_2 = \{ (q_0, \epsilon)\}$ 
of a transducer~$\transducer_2$.
Both are wrong since transducer~\transducer intended an initial
state space splitting with different output languages.
Transducer~$\transducer_2$ does not have the transduction~$(\epsilon, \{ \epsilon \}, 
\sem{\outputword_1 \join \outputword_2}_\outputcomp) \in \transductions(\transducer_2)$.
The transductions of~$\transducer_1$ are not equal to those of~\transducer,
since $\transductions(\transducer_1) = \{ 
    (\epsilon, \{ \epsilon \}, \sem{\outputword_1 \join \outputword_2}_\outputcomp),
    (\langle a \rangle, \{ \epsilon \}, \sem{(\outputword_1 \join \outputword_2) \concat \outputword_3}_\outputcomp),
    (\langle b \rangle, \{ \epsilon \}, \sem{(\outputword_1 \join \outputword_2) \concat \outputword_4}_\outputcomp)
\} \neq \transductions(\transducer)$.
\end{proof}

\begin{proposition}
An abstract transducer needs a set of initial transducer states
to allow for an elimination of $\epsilon$-moves.
That is, a set of initial transducer states with~$|\atmtinit| = 1$
is not sufficient for all $\epsilon$-input-free transducers
while maintaining their semantics.
\end{proposition}

\begin{proof}
Implication of the proof for proposition~\ref{prop:det}.
\end{proof}

%
%

%

%
%
%

%

\section{Transducer Abstraction}
\label{concept:atmtabstraction}%
\label{concept:at:abstraction}%
\newcommand{\tprecision}{\ensuremath{\precision}\xspace}%

Abstracting (widening) an abstract transducer is a means to 
provide its output for a larger set of input words, that is, a 
\emph{mechanism to increase sharing and with it the potential of reuse}.
That is, we explicitly rely on the fact that abstracting 
an automaton can widen its input language, 
and introduces non-determinism~\cite{AvniKupferman2013}.
We discuss different types of abstractions that are relevant 
for this work---Fig.~\ref{fig:transducers:abstractions} provides 
examples for abstractions.
Approaches for abstracting classical automata and symbolic automata 
have been presented in the past~\cite{BultanEtal2017,PredaEtal2015}, 
which can also be adopted for abstract transducers.
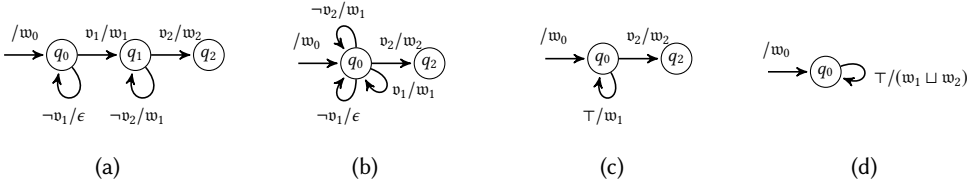
\begin{figure}[tp]
\begin{subfigure}[b]{.24\linewidth}
\begin{center}
\begin{tikzpicture}[node distance=8mm, scale=0.7, transform shape]
\node[] at (0,2) (top) {};
\node[] at (0,-1) (bot) {};
\node[] (qs) {};
\node[cfastate,right=of qs] (q0) {$q_0$};
\node[cfastate,right=of q0] (q1) {$q_1$};
\node[cfastate,right=of q1] (q2) {$q_2$};
\path[trans] (qs) edge node [pos=.5,label=above:{$/\outputword_0$}] {} (q0);
\path[trans] (q0) edge node [pos=.7,label=above:{$\abstinputword_1 / \outputword_1$}] {} (q1);
\path[trans] (q0) edge [out=310,in=260,looseness=8] node [pos=.7,label=below:{$\lnot \abstinputword_1 / \epsilon$}] {} (q0);
\path[trans] (q1) edge node [pos=.7,label=above:{$\abstinputword_2 / \outputword_2$}] {} (q2);
\path[trans] (q1) edge [out=310,in=260,looseness=8] node [pos=.7,label=below:{$\lnot \abstinputword_2 / \outputword_1$}] {} (q1);
\end{tikzpicture}
\end{center}
\caption{}
\end{subfigure}
\begin{subfigure}[b]{.24\linewidth}
\begin{center}
\begin{tikzpicture}[node distance=8mm, scale=0.7, transform shape]
\node[] at (0,2) (top) {};
\node[] at (0,-1) (bot) {};
\node[] (qs) {};
\node[cfastate,right=of qs] (q0) {$q_0$};
\node[cfastate,right=of q0] (q2) {$q_2$};
\path[trans] (qs) edge node [pos=.2,label=above:{$/\outputword_0$}] {} (q0);
\path[trans] (q0) edge [out=340,in=300,looseness=8] node [pos=.7,label=right:{$\abstinputword_1/ \outputword_1$}] {} (q0);
\path[trans] (q0) edge [out=270,in=230,looseness=8] node [pos=.7,label=below:{$\lnot \abstinputword_1 / \epsilon$}] {} (q0);
\path[trans] (q0) edge [out=90,in=130,looseness=8] node [pos=.7,label=above:{$\lnot \abstinputword_2 / \outputword_1$}] {} (q0);
\path[trans] (q0) edge node [pos=.7,label=above:{$\abstinputword_2 / \outputword_2$}] {} (q2);
\end{tikzpicture}
\end{center}
\caption{}
\end{subfigure}
\begin{subfigure}[b]{.22\linewidth}
\begin{center}
\begin{tikzpicture}[node distance=8mm, scale=0.7, transform shape]
\node[] at (0,2) (top) {};
\node[] at (0,-1) (bot) {};
\node[] (qs) {};
\node[cfastate,right=of qs] (q0) {$q_0$};
\node[cfastate,right=of q0] (q2) {$q_2$};
\path[trans] (qs) edge node [pos=.2,label=above:{$/\outputword_0$}] {} (q0);
\path[trans] (q0) edge [out=300,in=260,looseness=8] node [pos=.7,label=below:{$\top / \outputword_1$}] {} (q0);
\path[trans] (q0) edge node [pos=.7,label=above:{$\abstinputword_2 / \outputword_2$}] {} (q2);
\end{tikzpicture}
\end{center}
\caption{}
\end{subfigure}
\begin{subfigure}[b]{.25\linewidth}
\begin{center}
\begin{tikzpicture}[node distance=8mm, scale=0.7, transform shape]
\node[] at (0,2) (top) {};
\node[] at (0,-1) (bot) {};
\node[] (qs) {};
\node[cfastate,right=of qs] (q0) {$q_0$};
\path[trans] (qs) edge node [pos=.2,label=above:{$/\outputword_0$}] {} (q0);
\path[trans] (q0) edge [out=20,in=340,looseness=8] node [pos=.7,label=right:{$\top / (\outputword_1 \: \artefactjoin \: \outputword_2)$}] {} (q0);
\end{tikzpicture}
\end{center}
\caption{}
\end{subfigure}
\caption{Examples for different types of abstractions. Abstractions are applied 
step-wise from left to right:
(a)~we start with the unabstracted transducer, 
(b)~we conduct a state abstraction by merging states~$q_0$ and $q_1$, 
(c)~we abstract the input alphabet, 
(d)~we abstract the output alphabet.}
\label{fig:transducers:abstractions}
\end{figure}

Given an abstract transducer~$\transducer$, the abstraction
operator~$\abstraction^\tprecision: \transducers \rightarrow \transducers$ with widening---with the abstraction precision~$\tprecision$ as 
an implicit parameter that determines the level of abstraction 
to achieve---has to guarantee that the resulting abstraction 
overapproximates both the set of transductions and the 
set of accepting transductions:
\begin{definition}[Overapproximation]
An abstract transducer~$\transducer_1$ \emph{overapproximates} 
another abstract transducer~$\transducer_2$, which we denote 
by~$\transducer_2 \models \transducer_1$, if and only if 
$\transducer_1$ overapproximates both the set of transductions
and the set of accepting transductions of transducer~$\transducer_2$,
that is, $\transducer_2 \models \transducer_1$
if and only if 
$\forall (\bar{\sigma}_2, \hat{\sigma}_2, \hat{\theta}_2) \in \acctransductions(\transducer_2): \, \exists (\bar{\sigma}_1, \hat{\sigma}_1, \hat{\theta}_1) \in \acctransductions(\transducer_1): \, \bar{\sigma}_1 = \bar{\sigma}_2 \land \hat{\sigma}_1 = \hat{\sigma}_2 \land \hat{\theta}_2 \subseteq_C \hat{\theta}_1$
and $\forall (\bar{\sigma}_2, \hat{\sigma}_2, \hat{\theta}_2) \in \transductions(\transducer_2): \, \exists (\bar{\sigma}_1, \hat{\sigma}_1, \hat{\theta}_1) \in \transductions(\transducer_1): \, \bar{\sigma}_1 = \bar{\sigma}_2 \land \hat{\sigma}_1 = \hat{\sigma}_2 \land \hat{\theta}_2 \subseteq_C \hat{\theta}_1$.
The relation~$\subseteq_C$ denotes the inclusion relation of 
the concrete language lattice of the output language domain.
\end{definition}

\subsection{State Abstraction}
The classical approach to abstract an automaton is \emph{state 
abstraction}, that is, to merge several control states into 
one~\cite{PinchinatMarchand2000}.
Please note that this approach can also be used for abstracting 
output closures, which is the case if control states within 
an $\epsilon$-closure are merged:
\nopagebreak[4]
\newcommand{\mergestates}{\ensuremath{\textsf{qmerge}}\xspace}%
\begin{definition}[Control State Merge]A \emph{state merge} for a given abstract transducer~$\transducer$ 
is conducted by merging a set of its control states 
$\atmtstates_m \subseteq \atmtstates$ into into one new 
state~$q_m$, and results in a new abstract 
transducer~$\transducer_m$.
We denote this process by the operator~$\mergestates: \transducers \times 2^\atmtstates \rightarrow \transducers$, that is, $\transducer_m = \mergestates(\transducer, \atmtstates_m)$.
The actual definition of operator~\mergestates is given by Algorithm~\ref{alg:qmerge}.
\end{definition}
\begin{algorithm}[t]
\begin{small}
\caption{$\mergestates{}(\transducer, Q_m)$}
\label{alg:qmerge}
\begin{algorithmic}[1]
\vspace{1mm}
\INPUT Abstract transducer~$\transducer = (\atmtstates, 
    \inputdomain, \outputdomain, \atmtinit, \atmtfinal, \atmtrelation)$, \\
    \hspace{5mm} set $Q_m$ of states to merge
\vspace{1mm}
\OUTPUT Abstract transducer~$\transducer'$, with~$\transducer \models \transducer'$
\vspace{1mm}
\VARDECL Control state~$q_m$ that is not in the set $\atmtstates$ of transducer \transducer
\vspace{1mm}
\newline
\COMMENT{// Define the abstraction $\alpha$}
\STATE $\alpha = \{ \; (q, q') \; | \; q \in \atmtstates \land
    q' = q_m \, \text{if} \, q \in Q_m \, \text{else} \, q' = q \; \} $
\newline
\COMMENT{// New set of control states}
\STATE $\atmtstates' = (\atmtstates \setminus Q_m) \cup \{q_m\}$
\newline
\COMMENT{// New set of accepting states}
\STATE $\atmtfinal' = \{ \; \alpha(q) \; | \; q \in \atmtfinal \; \}$
\newline
\COMMENT{// New initial transducer state}
\STATE $\atmtinit' = \imagejoin \{ \; (\alpha(q), M) \; | \; (q, M) \in \atmtinit \; \}$
\newline
\COMMENT{// New transition relation}
\STATE $\atmtrelation' = \{ \; (\alpha(q), \abstinputword, \alpha(q'), \outputword) \; | \; (q, \abstinputword, q', \outputword) \in \atmtrelation \;\} $
\newline
\COMMENT{// Compose the resulting transducer}
\STATE \textbf{return} $(\atmtstates', \inputdomain, 
    \outputdomain, \atmtinit', \atmtfinal', \atmtrelation')$
\end{algorithmic}
\vspace{1mm}
\end{small}
\end{algorithm}

\begin{proposition}
\label{prop:qmerge}
Given an abstract transducer~\transducer, a 
transformation~$\transducer' = \mergestates(\transducer, Q_m)$
results in a new abstract transducer~$\transducer'$,
with $\transducer \models \transducer'$, that is, transducer~$\transducer'$
overapproximates transducer~$\transducer$.
\end{proposition}

\begin{proof}
We have to show that (1)~each 
input~$(\inputword, \inputwordset) \in \inputalphabet^* \times 2^{\inputalphabet^*}$ 
that leads to a feasible run~$\transducerstate = \runfromset_\transducer(\inputword, \inputwordset)$ on~\transducer also leads to a feasible 
run $\transducerstate' = \runfromset_{\transducer'}(\inputword, \inputwordset)$ 
on transducer~$\transducer'$, and for each 
element~$(q, \outputword) \in \transducerstate$ 
there exists an element $(q', \outputword') \in \transducerstate'$, 
with $\sem{\outputword} \subseteq \sem{\outputword'}$.
Furthermore, we have to show that (2)~each input that leads to 
an accepting run on transducer~$\transducer$ also lead to an 
accepting run on transducer~$\transducer'$.
Given a run $\run = \langle \atmtinit, \ldots, \transducerstate_n \rangle$
that is feasible on transducer~\transducer for a given 
input~$(\inputword, \inputwordset) \in \inputalphabet^* \times 2^{\inputalphabet^*}$.
The same input will also produce a feasible 
run~$\run' = \langle \atmtinit', \ldots, \transducerstate_n' \rangle$ on transducer~$\transducer'$.
For each~$\transducerstate_i \in \run$ with~$(q_1, \cdot) \in \transducerstate_i$
or~$(q_2, \cdot) \in \transducerstate_i$, the corresponding transducer
state~$\transducerstate_i' \in \run'$ will contain the merged
control state~$q_m$ with a corresponding abstract output word, that is,
$(q_m, \outputword) \in \transducerstate_i'$.
The definition of~$\mergestates$ ensures that all transitions
from either control state~$q_1$ or~$q_2$ are also possible from control state~$q_m$:
All transitions in $\atmtrelation$ from or to a control 
state in $\atmtstates_m$ are replaced by corresponding transitions
from or to control state~$q_m$.
In case a control state to merge is included in the initial transducer
state~$\atmtinit$, it is replaced by control state~$q_m$ in the
initial transducer state~$\atmtinit'$ of transducer~$\transducer'$.
The non-deterministic nature of abstract transducers ensures that
all transitions that match will also be taken: One control state
can have a set of successor states for a given input.
The transformation of the set of accepting control states~$\atmtfinal$
to the set $\atmtfinal'$ ensures that if one of the states to merge
was an accepting state, also state~$q_m$ will become an accepting state;
states in~$\atmtfinal$ that are not included in~$\atmtstates_m$ 
stay accepting states in~$\atmtfinal'$.
That is, all transitions---and runs on them---that were possible 
from or to control states in the set~$\atmtstates_m$ are still 
possible (lead feasible or accepting runs) in the new abstract 
transducer~$\transducer'$, but now start or end in control state~$q_m$.
\end{proof}

\noindent
Please note that abstracting abstract transducers by merging control 
states does neither affect the number of transitions nor their 
labeling---both the input symbols and the output symbols on transitions 
stay the same, but output languages of epsilon closures can change.

\begin{definition}[State Abstraction]The state abstraction~$\abst{\transducer}_\atmtstates^\tprecision$ 
of an abstract transducer~$\transducer$ results in a new abstract 
transducer~$\transducer'$ that is computed based on an abstraction 
precision~\tprecision, with~$\transducer \models \transducer'$.
The abstraction precision determines which states to keep separated 
and which to combine into one state---which represents the 
corresponding equivalence class.
The abstraction precision~$\tprecision = \langle Q_1, \ldots, Q_n \rangle$
defines a list of disjoint sets of control states that should be combined.
A state abstraction is conducted as follows:
$$
\abst{\transducer}_\atmtstates^\tprecision = \begin{dcases*}
    \mergestates(\transducer, Q_1) 
        & \text{if} $\tprecision = \langle Q_1 \rangle$ \\
    \mergestates(\abst{\transducer}_\atmtstates^{\langle Q_2, \ldots, Q_n \rangle}, Q_1) 
        & \text{if} $|\tprecision| > 1$ \text{and} $\tprecision = \langle Q_1, Q_2, \ldots \rangle$ \\
\end{dcases*}
$$
\end{definition}

\subsection{Input Alphabet Abstraction}
An abstraction approach that influences the abstract input words
of the transitions is input alphabet abstraction, which is the 
process of changing the abstract input word~$\abstinputword$ of a transducer 
transition~$\tau = (q, \abstinputword, q', \outputword) \in \atmtrelation$
to an new abstract input word~$\abstinputword'$, 
with $\sem{\abstinputword}_\inputcomp \subseteq \sem{\abstinputword'}_\inputcomp$:

\begin{definition}[Input Alphabet Abstraction]An \emph{input alphabet abstraction}~$\abst{\transducer}_I^\tprecision$
of an abstract transducer~\transducer results in a new abstract 
transducer were some of the abstract input words
of its control transitions were widened based on the given
abstraction precision~$\tprecision \in \precisions_I$.
The abstraction precision~$\tprecision$ for
input alphabet abstraction maps an abstraction precision~$\precision_\inputcomp$
that is applicable to the abstract input domain to each
of the transducer's control transitions, that is, it 
is a left-total function~$\precision: \Delta \rightarrow \precisions_\inputcomp$.
The result is an abstract transducer with a widened transition relation:
$$\atmtrelation' = \{ \; (q, \abst{\abstinputword}_\inputcomp^{\precision_\inputcomp}, q', \outputword) 
\; | \; \tau = (q, \abstinputword, q', \outputword) \in \atmtrelation 
\land (\tau, \precision_\inputcomp) \in \precision \; \}.$$
\end{definition}

\subsection{Output Alphabet Abstraction}
Along with this work, we introduce an output alphabet abstraction,
which adjusts the abstract output words of transitions.
It denotes the process of changing the abstract output word~$\outputword$ 
of a transducer transition~$\tau = (q, \abstinputword, q', \outputword) \in \atmtrelation$
to an new abstract output word~$\outputword'$, with $\sem{\outputword}_\outputcomp \subseteq \sem{\outputword'}_\outputcomp$:
\nopagebreak[4]
\begin{definition}[Output Alphabet Abstraction]An \emph{output alphabet abstraction}~$\abst{\transducer}_O^\tprecision$
of an abstract transducer~\transducer
results in a new transducer were some of the abstract output words
of its control transitions were widened based on the given
abstraction precision~$\tprecision \in \precisions_O$.
The precision~$\tprecision$ for
output alphabet abstraction maps an abstraction precision~$\precision_\outputcomp$
that is applicable to the output domain to each
of the transducer's transitions, that is, it 
is a left-total function~$\precision: \Delta \rightarrow \precisions_\outputcomp$.
The result is an abstract transducer with a widened transition relation:
$$\atmtrelation' = \{ \; (q, \abstinputword, q', \abst{\outputword}_\outputcomp^{\precision_\outputcomp}) 
\; | \; \tau = (q, \abstinputword, q', \outputword) \in \atmtrelation 
\land (\tau, \precision_\outputcomp) \in \precision \; \}.$$
\end{definition}

\noindent
Please note that also the computation of the abstract output 
closure---see Sect.~\ref{concept:abstepsclosure}---yields a
form of output alphabet abstraction.

\section{Transducer Reduction}
\label{concept:at:reduce}%
\label{sec:at:reduce}%
Besides abstraction techniques, also techniques for the reduction 
of abstract transducers are important. 
Such techniques help to reduce the number of control states, 
the number of control transitions, and the degree of non-determinism 
of a given abstract transducer.
That is, they help to reduce the costs of using and running 
abstract transducers for particular inputs, for example, to 
conduct a verification task.
Minimization is related to reduction but aims at ending up in 
finite state machines with a minimal number of states---an optimum.

The number of control states of an abstract transducer is 
critical for the performance of its use in an analysis procedure. 
Since a minimization is too 
expensive~\cite{JiangRavikumar1993,DAntoniVeanes2014,BjoerklundMartens2012}, 
we propose to adopt reduction techniques as known for NFAs 
to reduce the size and the degree of non-determinism of abstract 
transducers---a low degree of non-determinism is critical for 
efficient execution of non-deterministic finite state machines~\cite{KoHan2014}.

Abstract transducers can be reduced by merging control states,
or their transitions, as long as the set of transductions 
and the set of accepting transductions is preserved.
Please note that we assume, if not stated otherwise, that 
$\epsilon$-moves were removed \emph{before} applying the 
reduction techniques that we describe here.

\newcommand{\reduceop}{\ensuremath{\textsf{reduce}}\xspace}%
\begin{definition}[Operator~\reduceop]The (generic) reduction operator~$\reduceop: \transducers \rightarrow \transducers$
reduces a given abstract transducer~\transducer.
Instances of this operator have to guarantee to produce an equivalent 
abstract transducer, that is, $\transducer \equiv \reduceop(\transducer)$.
\end{definition}

\subsection{Reduction by State Merging}
Before we continue to outline an algorithm for reducing abstract
transducers by merging control states, we provide more definitions:

\begin{definition}[Control State Equivalence]
Two control states~$q_1, q_2 \in \atmtstates$ of an abstract 
transducer~\transducer are called \emph{equivalent} to each other,
that is, $q_1 \equiv q_2$ , if and only if they can be merged 
without affecting the transducer`s set of transductions nor 
its set of accepting transductions, that is, if and only 
if~$\transducer \equiv \mergestates(\transducer, \{q_1, q_2\})$.
\end{definition}

\noindent
Based on the definition of control state equality, we define
the equality of abstract transducer states:

\begin{definition}[Transducer State Equivalence]Two transducer states~$\transducerstate_1, \transducerstate_2 \in \transducerstates$
are called \emph{equivalent} if and only if they describe equivalent
pairs of control states and abstract output words, that is,
if and only if 
$\forall (q, \outputword) \in \transducerstate_1: \exists (q', \outputword') \in \transducerstate_2: q \equiv q' \land \outputword \equiv \outputword'$
and
$\forall (q, \outputword) \in \transducerstate_2: \exists (q', \outputword') \in \transducerstate_1: q \equiv q' \land \outputword \equiv \outputword'$.
\end{definition}

\noindent
To determine whether merging two control states maintains the
set of transductions, the notion of left transductions is essential:
\nopagebreak[4]
\newcommand{\lefttransductions}{\ensuremath{\overleftarrow{\transductions}}\xspace}%
\begin{definition}[Left Transductions]The set of~\emph{left transductions}
$\lefttransductions(\transducer, q) \subseteq \inputalphabet^* \times 2^{\inputalphabet^*} 
\times 2^{\outputalphabet^\infty}$
to a given control state~$q \in \atmtstates$, which belongs to 
a particular abstract transducer~$\transducer \in \transducers$,
is the set of all transductions that can be produced on paths
that start in the initial transducer state~$\atmtinit$ and that
\emph{reach} the given control state~$q$ with a feasible run:
\begin{align*}
\lefttransductions(\transducer, q) = \bigcup \; \{ \; 
    & (\inputword, \inputwordset, \sem{\outputword'}_\outputcomp) \; \\
    & | \; (q, \outputword') \in \runfromset_\transducer(\inputword, \inputwordset) \\
    & \land \inputword \in \inputalphabet^* 
        \land \inputwordset \subseteq \inputalphabet^* 
        \land \outputword' \neq \bot \; \}.
\end{align*}
\end{definition}

\begin{proposition}
A transformation~$\transducer' = \mergestates(\transducer, \{q_1, q_2\})$
maintains both the set of transductions and the set of accepting 
transductions if the left-transductions of the control states~$q_1$ and $q_2$ are 
equal, that is, $\transducer' \equiv \mergestates(\transducer, \{q_1, q_2\})$
if $\lefttransductions(\transducer, q_1) = \lefttransductions(\transducer, q_2)$.
\end{proposition}

\begin{proof}
Control state~$q_1$ is reachable by runs that correspond to the
set of left transductions~$\lefttransductions(\transducer, q_1)$
and control state~$q_2$ by runs that correspond to the set
of left transductions~$\lefttransductions(\transducer, q_2)$.
The proposition states that if we merge control states~$q_1$ and~$q_2$,
with~$\lefttransductions(\transducer, q_1) = \lefttransductions(\transducer, q_2)$
into a new state~$q_m$ of a new 
transducer~$\transducer' = \mergestates(\transducer, \{ q_1, q_2\})$
then this transducer is equivalent~$\transducer' \equiv \transducer$
to the original one. 
(1)~First, we show that control state~$q_m$ is reachable by all feasible 
runs that can also reach control state~$q_1$ or $q_2$, and that there is 
no feasible run that can reach~$q_m$ but neither state $q_1$ or $q_2$.
That is, we show that~$\lefttransductions(\transducer', q_m) = \lefttransductions(\transducer, q_1) = \lefttransductions(\transducer, q_2)$:
The operation~\mergestates ensures that all transitions that entered
either state~$q_1$ or state~$q_2$ also enter state~$q_m$; that is,
all feasible runs that reached~$q_1$ or~$q_2$ now reach state~$q_m$
and since~$q_m$ is a new state it is only reachable by these runs.
(2)~Next, we show that all runs that are feasible from control state~$q_1$ 
or~$q_2$ are also feasible from control state~$q_m$, and there is 
no feasible run from state~$q_m$ that is not feasible from control 
state~$q_1$ or~$q_2$:
The construction process of~$q_m$ ensures that all transitions
that leave states~$q_1$ or $q_2$ also leave state~$q_m$, and
no other transitions get added to leave this state; that is,
all feasible runs that start in control state~$q_m$ are also
feasible runs if they start in control state~$q_1$ or $q_2$.
(3)~Finally, we have to show that all runs that are accepting from
control state~$q_1$ or $q_2$ are also accepting from state~$q_m$,
and there is no accepting run from control state~$q_m$ that 
is not accepting from state~$q_1$ or $q_2$:
The operator~\mergestates merges states~$q_1$ and $q_2$ into a state~$q_m$,
which becomes an accepting control state if also state
$q_1$ or state~$q_2$ is an accepting control state.
That is, the inputs~$\{ (\bar{\sigma}, \hat{\sigma}) \, 
| \, (\bar{\sigma}, \hat{\sigma}, \cdot) \in \lefttransductions(\transducer, q_1)$
become elements of set of accepting transductions of transducer~$\transducer'$
if they were also accepted by transducer~$\transducer$.
All inputs that get accepted by runs starting from control state~$q_1$
or state~$q_2$, get also accepted by runs that start from control state~$q_m$.
\end{proof}

\noindent
Statements about the result of manipulating an abstract transducer 
by merging control states are also possible based on the 
notion of right transductions:
\nopagebreak[4]
\newcommand{\righttransductions}{\ensuremath{\overrightarrow{\transductions}}\xspace}%
\begin{definition}[Right Transductions]The set of~\emph{right transductions}
$\righttransductions(\transducer, q, \outputword_0) \subseteq \inputalphabet^* \times 2^{\inputalphabet^*} 
\times 2^{\outputalphabet^\infty}$
of a given control state~$q \in \atmtstates$, which belongs to a 
specific abstract transducer~$\transducer \in \transducers$,
with initial abstract output word~$\outputword_0$,
is the set of all transductions that can be produced on the 
\emph{feasible runs} that start from the given transducer state~$(q, \outputword_0)$:
\begin{align*}
\righttransductions(\transducer, q, \outputword_0) = \bigcup \; \{ \; 
    & (\bar{\sigma}, \hat{\sigma}, \sem{\outputword'}_\outputcomp) \; \\
    & | \; (\cdot, \outputword') \in \runfromset_\transducer(\{ (q, \outputword_0) \}, 
    \bar{\sigma}, \hat{\sigma}) \\
    & \land \bar{\sigma} \in \inputalphabet^* 
        \land \hat{\sigma} \subseteq \inputalphabet^* 
        \land \outputword' \neq \bot \; \}.
\end{align*}
\end{definition}

\newcommand{\rightlangacc}{\ensuremath{\overrightarrow{\langacc}}\xspace}%
\begin{definition}[Right Accepted Language]The \emph{right accepted language} of a given abstract 
transducer~$\transducer$
for a given control state~$q$ is the set of 
pairs~$(\bar{\sigma}, \hat{\sigma}) \in \inputalphabet^* \times 2^{\inputalphabet^*}$ 
that lead to an accepting run if started from the given control state~$q$:
\begin{align*}
\rightlangacc(\transducer, q) = 
\{ \; & (\inputword, \hat{\sigma}) \in \langin(\transducer) \; \\
    & | \; (q', \cdot) \in \runfromset_\transducer(\{(q, \outputword_\epsilon)\}, \inputword, \hat{\sigma}) \\
    & \land q' \in \atmtfinal \; \}.
\end{align*}
\end{definition}

\begin{proposition}
Merging two control states~$q_1, q_2 \in \atmtstates$ of an 
abstract transducer~$\transducer$, which results in a new abstract 
transducer, maintains the set of transductions if their sets 
of \emph{right transductions} are equal, that is, 
$\transductions(\transducer) = \transductions(\mergestates(\transducer, \{q_1, q_2 \}))$ 
if $\righttransductions(\transducer, q_1) = \righttransductions(\transducer, q_2)$.
Please note that we do not make a proposition about the set of
\emph{accepted} transductions here.
\end{proposition}

\begin{proof}
Let the set of left transductions of two control states~$q_1$ 
and $q_2$ be different to each other, that is, 
$\lefttransductions(\transducer, q_1) \neq \lefttransductions(\transducer, q_2)$.
From proposition~$\ref{prop:qmerge}$ and the corresponding proof we known
that a merge of control states~$q_1$ and~$q_2$ leads to an overapproximation,
that is, $\transducer \models \mergestates(\transducer,\{q_1, q_2\})$.
It remains to be shown that the set of transductions is preserved if the 
right-transductions of two control states to merge are actually equal:
$\transductions(\transducer) = \transductions(\mergestates(\transducer, \{q_1, q_2\}))$
if $\righttransductions(\transducer, q_1) = \righttransductions(\transducer, q_2)$,
that is, that the merge does not add additional transductions.
To add additional transductions it would be necessary that the
set of right-transductions of control state~$q_m$ overapproximates
the union of the right-transductions of control states~$q_1$ and~$q_2$.
Nevertheless, since~$\righttransductions(\transducer, q_1)$ is
    equivalent to~$\righttransductions(\transducer, q_2)$ also
$\righttransductions(\transducer, q_m)$ does not add additional
right transductions, that is, $\righttransductions(\transducer, q_1) = \righttransductions(\transducer, q_2) = \righttransductions(\transducer, q_m)$.
\end{proof}

\begin{proposition}
Merging two control states~$q_1, q_2 \in \atmtstates$ of an abstract
transducer~$\transducer$, which results in a new transducer, 
does \emph{not} maintain the set of \emph{accepting transductions} 
if their sets of left transductions are not equal to each other.
That is, $\acctransductions(\transducer) \neq \acctransductions(\mergestates(\transducer, \{q_1, q_2\}))$
if~$\lefttransductions(\transducer, q_1) \neq \lefttransductions(\transducer, q_2)$.
\end{proposition}

\begin{proof}
Let $q_1$ and $q_2$ be two control states of an abstract 
transducer~$\transducer$, with~$q_1 \in \atmtfinal$ and~$q_2 \not \in \atmtfinal$.
Merging these states by~$\mergestates(\transducer, \{ q_1, q_2 \})$
results in a new transducer~$\transducer'$ with a control state~$q_m$
into that $q_1$ and $q_2$ have been merged, and that became
an accepting control state~$q_m \in \atmtfinal'$.
In case the left transductions~$\lefttransductions(\transducer, q_1)$
and $\lefttransductions(\transducer, q_2)$ are different to each other,
different inputs can reach states~$q_1$ and~$q_2$. 
Both inputs that reached~$q_1$ or~$q_2$ can reach the 
control state~$q_m$, and all these inputs now result
in accepting runs since~$q_m \in \atmtfinal'$, that is,
also runs for inputs that reached~$q_2$ and that were not 
accepting before now reach the accepting control state~$q_m$,
resulting in an overapproximation of the set of accepting transductions.
\end{proof}

\begin{definition}[Left Equivalent]The \emph{left~equivalence relation}~$\equiv_L \subseteq \atmtstates \times \atmtstates$
describes the pairs of control states that are equivalent to 
each other and that have the same set of left-transductions---%
it is a subset of control state equivalence relation.
That is, $(q_1, q_2) \in \equiv_L$ if $q_1 \equiv q_2$ and $\lefttransductions(\transducer, q_1) = \lefttransductions(\transducer, q_2)$.
\end{definition}

\newcommand{\enteringtrans}{\ensuremath{\textsf{entering}}\xspace}
\begin{proposition}
Given an input-$\epsilon$-free abstract transducer~\transducer,
a set of two control states~$Q_m = \{q_1, q_2 \} \subseteq \atmtstates$ 
of transducer~\transducer 
satisfy~$\lefttransductions(q_1) = \lefttransductions(q_2)$ if
$\forall \, (q, \abstinputword, q', \outputword) \in \enteringtrans(Q_m): 
    \forall \, (q'', \abstinputword', q', \outputword') \in \enteringtrans(Q_m): 
    \abstinputword \equiv \abstinputword' \land \outputword \equiv \outputword'
    \land \lefttransductions(q) = \lefttransductions(q'')$.
We use the auxiliary function~$\enteringtrans(Q) = \{ (q, \abstinputword, q', \outputword) \in \atmtrelation \,  | \, q' \in Q \}$.
\end{proposition}

\begin{proof}
Given the set of all control states~$Q_p \subseteq Q$ from these
control states in~$Q_m = \{ q_1, q_2 \}$ are directly reachable, 
that is,
$Q_p = \{ q \, | \, (q, \cdot, q', \cdot) \in \atmtrelation \land q' \in Q_m \} $.
If all states in the set~$Q_p$ have the same set of left-transductions,
then only transitions from control states in the set~$Q_p$ 
to those control states in the set~$Q_m$ can affect whether
or not the sets of left-transductions of states in~$Q_m$ are 
not equal to each other.
\end{proof}

\newcommand{\reduceleftop}{\ensuremath{\reduceop_\text{Left}}\xspace}%
\begin{definition}[Operator~$\reduceleftop$]The reduction operator~$\reduceleftop: \transducers \rightarrow \transducers$
reduces a given abstract transducer~$\transducer$ by merging all 
control-states that are left-equivalent to another.
The transformation satisfies~$\reduceleftop(\transducer) \equiv \transducer$.
\end{definition}




\noindent
Existing algorithms~\cite{DAntoniVeanes2014,LantoniVeanes2017} for 
reducing automata and symbolic transducers are not applicable because 
the set of transductions is not taken into account in the definition of equivalence.

\section{Abstract Transducer Analysis}
\label{sec:artefact:transducer:analysis}
We now present a generic and configurable program analysis that 
\emph{executes an abstract transducer}. 
This abstract transducer analysis keeps track of the \emph{current 
transducer state} while processing the input. 
The analysis can be configured, for example, to determine the 
extent to which the transducer states should be tracked in a 
path sensitive manner---path sensitivity might be needed for 
particular analysis purposes only. 
Thus, we can mitigate the state-space explosion problem in some cases. 
The transducer analysis is the foundation for several analyses 
that we describe in this work, for example, for the 
\yarn transducer analysis, and the precision transducer analysis.

\subsection{Abstract Transducer CPA}
\label{sec:at:cpa}

Our abstract transducer analysis is built on the concept of 
configurable program analysis~(CPA)~\cite{CPA,CPAplus}.
The abstract transducer CPA
$$\cpa_\symtransducer = (\abstdomain_\symtransducer, 
\transabs{}{}_\symtransducer, \strengthen_\symtransducer, \mergeop_\symtransducer, \stopop_\symtransducer, \precop_\symtransducer, \targetop_\symtransducer)$$
tracks a set of states of a given abstract transducer~$\transducer =  
(\atmtstates, \inputdomain, \outputdomain, \atmtinit, \atmtfinal, \atmtrelation)$.
The CPAs behavior is configured by using different variants of its operators.
For example, varying the operator~$\mergeop_\transducer$ can 
configure the analysis to operate path sensitive, or only context 
sensitive and flow sensitive~\cite{CPA}.
We rely on the strengthening operator~$\strengthen_\transducer$ for 
instantiating parameterized outputs.
Other program analyses, which run in parallel to the abstract transducer 
analysis, can read and use use the output words for different purposes.
The abstract transducer analysis~$\cpa_\transducer$ is composed
of the following components:

\begin{component}{Abstract Domain $\abstdomain_\transducer$}
The abstract domain~$\abstdomain_\transducer = (\concrete, 
\lattice, \sem{\cdot}, \abst{\cdot})$ is defined based on a \emph{map 
lattice}~$\lattice = (\transducerstates, \top, \bot, \sqsubseteq, \join, \meet)$, 
with~$\transducerstates = 2^{\atmtstates \rightarrow \outputwords}$,
where each element~$\transducerstate \in \transducerstates$ of the 
lattice is an abstract transducer state.
One transducer state~$\transducerstate = \{(q, \outputword), \ldots\} \in \transducerstates$ 
is a mapping~$\transducerstate: \atmtstates \rightarrow \outputwords$
from control states to abstract output words.
The analysis starts with the initial transducer state~$\atmtinit$
of the abstract transducer to conduct runs for.
\end{component}

\newcommand{\look}{\ensuremath{\textsf{look}}\xspace}%
\begin{component}{Transfer Relation~$\transabs{}{}_\transducer$}
The transfer relation~$\transabs{}{}_\transducer \subseteq \transducerstates \times \controlflows \times \transducerstates \times \precisions$ 
defines abstract successor states of an abstract 
state~$\transducerstate = \{(q, \outputword), \ldots \} \in \transducerstates$ 
for a given control-flow transition~$g \in \controlflows$ 
and abstraction precision~$\precision \in \precisions$.
We define this transition relation \emph{without implicit stuttering}, 
that is, if there should be stuttering, 
the transducer must have corresponding transitions.
The transfer relation is defined as follows:
\begin{align*}
\transducerstate \transabs{g}{}_\transducer \{ \; \{ (q, \outputword) \} \; 
& | \;  (\abstinputword, \abstinputword_\lookahead) = \look(g, \lookahead) \\
& \land (q, \outputword) \in \runfromset_\transducer(\transducerstate, \abstinputword, \abstinputword_\lookahead) \\
& \land q \neq \qbot \; \}. 
\end{align*}
\noindent
Please note that the function~$\runfromset$ is implicitly
parameterized with an abstract closure operator~\outclosure.
The operator $\look: \controlflows \times \Natsz \rightarrow \abstinputwords \times \abstinputwords$
maps the given control-flow transition~$g \in \controlflows$ to
an abstract input word~$\abstinputword$ and provides a bounded 
lookahead of length~$\lookahead$ in form of the abstract input word~$\abstinputword_\lookahead$ which is derived from the 
control-flow transitions that follow transition~$g$ on the 
control-transition relation of the underlying analysis task.

\recap{Alphabet Translation}
The operator~$\look$ does not only provide the lookahead but
also translates between the alphabet of the graph that is traversed
to the abstract input alphabet of the abstract transducer.
That is, varying this operator provides different \emph{views} 
on the given input, for example, a control-flow transition~$g \in \controlflows$
can be translated to the function to that the transition belongs to,
or to the successor control location that is reached by 
the control transition.

The operator \mergeop can decide later if states should be
tracked separately or not.
\end{component}

\begin{operator}{$\strengthen_\transducer$}
The strengthening~\cite{CPA} operator~$\strengthen_\transducer: E_\times \times E_\times \times \absttransducerstates \rightarrow \absttransducerstates$ 
is called after all analyses that run in parallel have provided an abstract 
successor state as components for the composite state~$e_\times = (e_1, \ldots, e_n) \in E_\times$. 
At this point, the strengthening operator can access the information 
that is present in any of the component states~$e_i \in e_\times$ 
and use them to strengthen its own (component) state. 
We instantiate parameterized output words during strengthening. 
Information of an analysis that runs in parallel can be used to 
support various instantiation and synthesis mechanisms.

The strengthening~$\absttransducerstate' = \{(q', \outputword')\} = \strengthen_\transducer(e_\times, e_\times', \absttransducerstate)$
is conducted for a given transducer state~$\absttransducerstate = \{(q, \outputword)\}$,
which is the result of conducting a transducer 
transition~$\tau = (q, \abstinputword, q', \outputword) \in \atmtrelation$ 
for an input~$(\bar{\sigma}, \hat{\sigma}) \in \inputalphabet^* \times 2^{\inputalphabet^*}$.
Beside the information that can be found in the composite 
states~$e_\times$ and $e_\times'$, also the values that were 
bounded to the parameters of the abstract input word~$\abstinputword$ 
can be taken into account to instantiate the abstract output word~$\outputword'$.
A consistent binding of parameters among different transitions,
that is, for the whole program trace---as this is used by some 
aspects and corresponding weavers~\cite{AllanEtal2005}---is not yet supported.
\end{operator}

\begin{operator}{$\mergeop_\transducer$}
The merge operator~$\mergeop_\transducer: \absttransducerstates \times \absttransducerstates \times \precisions \rightarrow \absttransducerstates$ 
controls if two transducer states should get combined,
or if they should be explored separately and separate the state space. 
The behavior of the operator can be controlled based on a given 
precision~$\precision \in \precisions$. 
The default is to always separate two different abstract states, that is,
$\mergeop_\transducer = \mergesepop$~\cite{CPA}, which ensures 
the path sensitivity of the analysis.
Please note that the abstract transducer analysis is typically one
of several analyses that run as components of a composite analysis: 
Even if the analysis would conduct a merge, other component 
analyses might signal not to do so.
\end{operator}

\begin{operator}{$\stopop_\transducer$}
The coverage check operator~$\stopop_\transducer: \absttransducerstates \times 2^\absttransducerstates \rightarrow \Bools$ decides whether a given abstract state is already covered 
by a state reached or not. 
As default, we use the inclusion relation of the lattice, that is,
$\stopop_\transducer = \stopsepop$.
\end{operator}

\begin{operator}{$\precop_\transducer$}
The precision adjustment operator~$\precop_\transducer$ could conduct further 
abstraction of a given abstract state. 
We do not abstract here: A call $\precop_\transducer(\absttransducerstate,\precision,\cdot)$ 
returns the pair~$(\absttransducerstate, \precision) \in \absttransducerstates \times \precisions$ 
without adjustments.
\end{operator}

\begin{operator}{$\targetop_\transducer$}
The \emph{target operator}~$\targetop_\transducer: \absttransducerstates \rightarrow 2^\properties$ determines
the set of properties for that a given abstract state is a target state.
Each \emph{property is a task concern}, 
that is, the set of properties~$\properties \subset \concerns$ 
is a subset of the set~$\concerns$ of task concerns.
We assume that there is only one 
transition~$\tau = (\cdot, \cdot, q', \cdot) \in \atmtrelation$ 
for each accepting control state~$q' \in \atmtfinal$.
\recap{Concern Map}
We rely on a function~$\concernmap: \Delta \rightarrow 2^\concerns$
that maps each transducer transition to a set of task concerns.
Given an abstract transducer state~$\absttransducerstate = \{ (q_1, \cdot), \ldots (q_n, \cdot) \} \in \absttransducerstates$, the operator returns:
\begin{align*}
  \targetop_\transducer(\absttransducerstate) = 
  \bigcup \; \{ \; & \concernmap(\transtrans) \; | \; (q, \cdot) \in \absttransducerstate \: \\
  & \land \: q \in \atmtfinal \\
  & \land \transtrans = (\cdot, \cdot, q, \cdot) \in \atmtrelation \; \}
\end{align*}
\end{operator}

\subsection{Analysis Configurations}
\label{sec:analysisconfigs}

By relying on the CPA framework~\cite{CPA}, the abstract transducer
analysis is equipped with an inherent notion of configurability, and
can be instantiated several times and in different ways within the framework,
to conduct an analysis task in the most efficient and effective manner.

\subsubsection{Transducer Composition}
It might be necessary to execute several abstract transducers 
in parallel along with the state space exploration for an analysis task.
Given a list~$\langle \transducer_1, \ldots, \transducer_n \rangle$ 
of abstract transducers to run, a list~$\langle \cpa_1, \ldots, \cpa_m \rangle$
of analyses, with~$n \geq m > 1$, has to be instantiated.
We assume that these transducers have the same abstract input domain
and the same abstract output domain, and consider the composition
of transducers with different abstract output domains to be future work. 
The first approach (\emph{separation}) is to instantiate 
one analysis for each abstract transducer ($m = n$), 
which fosters a clear separation of concerns.
Each of the~$m$ instantiated analyses adds one component to 
composite (product) state that is formed by the composite analysis;
the number of CPAs operators that are invoked transitively by the 
CPA algorithm increases.
An alternative approach (\emph{union}) is to construct the 
union~$\transducer_\cup = \transducer_1 \cup ... \cup \transducer_n$
of the transducers to run and to run this single transducer~$\transducer_\cup$
with one abstract transducer analysis.
Also \emph{hybrid} approaches can be taken, that is, construct
unions for subsets of the transducers, and run others separately.

\subsubsection{State-Set Composition}
One abstract state~$\transducerstate = \{ (q_1, \cdot), (q_2, \cdot), \ldots \} \in \transducerstates$ 
of the transducer analysis can contain several control 
states from the set~$\atmtstates$ of the abstract transducer to 
run for a given analysis task.
The number of control states per abstract state
can be controlled by the transducer analysis and its operators,
for example, the operator \mergeop, which decides whether or not
to explore two abstract states separately.
The decision to join two different control states into one state 
set of one abstract state of the transducer analysis can affect
the path sensitivity of the analysis, that is, if it is possible 
to determine the branch of the state space that has led to 
a given control state.

\section{Related Work}

Abstract transducers combine different concepts and techniques 
from formal methods, automata theory, domain theory, and abstract 
interpretation, to end up in a generic type of abstract machine. 
We discuss the related work based on the different concepts that 
can be found in abstract transducers and explain the relationship
and differences to existing work.

\subsubsection{Symbolic Alphabet}
An abstract transducer can use arbitrarily composed  
abstract domains to define both its input and the output; 
for the input domain, we require that its lattice is dual to a Boolean algebra.
We introduce a special class of abstract domain,
the abstract word domain, to describe words
of complex entities, such as program traces of concurrent systems.
Symbolic finite automata and transducers~\cite{DAntoniVeanes2017}
share the idea of using theories to describe sets of input and output symbols. 
Other types of automata describe their input symbols based on 
predicates~\cite{NoordGerdemann2001}
or as multi-valued input symbols~\cite{KupfermanLustig2007,GallJeannet2007}.
From the perspective of abstract transducers, trace partitioning 
domains~\cite{TracePartitoningDomain}, lattice 
automata~\cite{GallJeannet2007}, and regular expressions over
lattice-based alphabets~\cite{MidtgaardEtal2016} are instances
of abstract word domains. 

\subsubsection{Output Closure}
With abstract transducers, we also introduce means to deal with 
$\epsilon$-loops that are annotated with outputs, that is, to 
compute and use finite symbolic representations of outputs that 
potentially consist of exponentially many and infinitely long words.
Compared to existing work~\cite{OonishiEtal2009,DAntoniVeanes2017}, 
we also consider $\epsilon$-moves that lead to dead ends as relevant, 
handle them in our algorithms, and do not consider them as 
candidates for removal.



\subsubsection{Lookahead}
In each step of processing input, abstract transducers can 
conduct a lookahead on the remaining input to determine which 
transitions to take.
Several other types of abstract machines provide the capability
of lookaheads, for example, tree transducers, which had been 
extended to support regular lookaheads~\cite{Engelfriet1977},
or extended symbolic finite state machines~\cite{DAntoniVeanes2015}.
A labeling with words instead of letters is also conducted
in the case of generalized finite automata~\cite{Sipser}, but
they consume full words in a transition step---instead of 
just one letter as is the case for abstract transducers.



\subsubsection{Transducer Abstraction}
By defining both the input alphabet and the output alphabet
of abstract transducers based on an abstract domain, we can make 
use of the full range of abstraction mechanisms that were developed 
in the context of abstract interpretation 
for abstracting abstract transducers, that is, to widen their set
of transductions.
Approaches for abstracting classical automata have been presented
in the past~\cite{BultanEtal2017}. A more recent work~\cite{PredaEtal2015}
presented techniques for abstracting symbolic automata,
which could also be adopted to abstract transducers.
This work is the first that proposes to abstract a type of 
transducer for increasing its sharing, that is, widen the set 
of words for which particular outputs are produced.


\subsubsection{Running Transducers}
Running automata in parallel to a program analysis is an established 
concept in the fields of program analysis and 
verification~\cite{BLAST,BeyerEtal2004,MultiGoalReachability}.
Algorithmic aspects of how automata are executed, 
for example, how the current state of automata is represented in the state 
space of the analysis task, are in many cases~\cite{HeizmannEtal2013} 
not discussed further, while the performance implications can 
be dramatical.
Work in the context of configurable program analysis~\cite{ConditionalMC,Testification,CorrectnessWitnesses} 
is most transparent about this.

\subsubsection{Transducers for Analysis and Verification}
Transducers are widely used in the context of program analysis and verification.
They are used, for example, for synthesis~\cite{PnueliRosner1989},
to describe the input-output-relation of programs~\cite{Hamlet1978,YuBultanIbarra2011,PnueliRosner1989,BotbolEtal2017},
and for string manipulations~\cite{YuBultanIbarra2011,SymbolicTransducers}.
Automata that produce an output---which is then used in the analysis process---%
have been proposed in the form of assumption automata~\cite{ConditionalMC}
for conditional model checking, error witness automata that output strengthening conditions~\cite{Testification} to narrow down the state space of the analysis process,
and for correctness witnesses~\cite{CorrectnessWitnesses}.



\section{Summary}

This work has introduced abstract transducers, a type of 
abstract machines that map between an input language and an output
language while taking a lookahead into account.
In contrast to established finite-state transducers, abstract
transducers have a strong focus on the intermediate language
that they produce, which has several implications on the 
design of algorithms that operate on these machines.
Both the input alphabet and the output alphabet of abstract
transducers consist of \emph{abstract words}, where
one abstract word denotes a set of concrete words. 
Means for representing, constructing, and widening of
abstract words, and for describing their relationship,
are provided by the corresponding abstract word domain.
Building on these abstract alphabets allows 
for \emph{abstracting} these transducers.

We use techniques from abstract interpretation
as the foundation for our abstraction mechanisms.
The concept of abstract transducers enables several new applications:
We discussed applications in the context of sharing task artifacts
for reuse within program analysis tasks.

From the concept of abstract transducers, we instantiate the 
concept of \emph{task artifact transducers}, which generalize a 
group of finite-state machines that are used in the context of 
program analysis and verification for reproducing and sharing 
information.
These transducers provide information that contributes to an 
analysis task and its solution.
The underlying graph structure of finite-state transducers allows 
us to capture the structure of information, share it, 
and enable its reuse.
Task artifact transducers have several applications and we 
outlined some of them: \yarn transducers, which provide 
sequences of program operations to weave into a transition system, 
and precision transducers, which are a means to define
the level of abstraction for different parts of the state space.
Other applications of task artifact transducers can be found in 
the context of providing and checking verification evidence, 
for example, 
transducers for error witnesses, which provide information that 
guides towards specification violations, or transducers for 
correctness witnesses, which provides certificates to check
while traversing the control flow of programs.
\newpage

\bibliographystyle{ACM-Reference-Format}
\bibliography{related}

\end{document}